\documentclass[titlepage]{article}

\usepackage[pages=all, color=black, position={current page.south}, placement=bottom, scale=1, opacity=1, vshift=5mm]{background}

\usepackage[margin=1in]{geometry} 

\usepackage{amsmath}

\usepackage{amsthm}
\usepackage{amsfonts}
\usepackage{array}
\usepackage{algorithm}
\usepackage{algorithmicx}
\usepackage{comment}

\usepackage[utf8]{inputenc}
\usepackage{hyperref}
\hypersetup{
	unicode,
	pdfauthor={Author One, Author Two, Author Three},
	pdftitle={A simple article template},
	pdfsubject={A simple article template},
	pdfkeywords={article, template, simple},
	pdfproducer={LaTeX},
	pdfcreator={pdflatex}
}

\usepackage[sort&compress,numbers,square]{natbib}
\theoremstyle{plain}
\newtheorem{theorem}{Theorem}

\newtheorem{assumption}[theorem]{Assumption}
\newtheorem{proposition}[theorem]{Proposition}

\theoremstyle{definition}
\newtheorem{definition}[theorem]{Definition}

\newcommand{\CZ}[1]{{\color{blue} (\textbf{claire:} #1)}}

\usepackage{graphicx, color}
\graphicspath{{fig/}}
\usepackage{algorithm, algpseudocode} 
\usepackage{mathrsfs} 

\title{The Impact of Cut-ins on Mixed-Autonomy Traffic Flow: Bridging Microscopic Game-Theoretic Model and Macroscopic Flow Analysis}
\author{Yu (Fred) Song$^a$}

\date{
	$^a$ Department of Civil and Architectural Engineering and Construction Management \\
    University of Wyoming \\ 
    1000 E University Ave., Laramie, WY 82070, USA \\
    \texttt{ysong@uwyo.edu}%
}

\begin{document}
\maketitle

\begin{abstract}
The transition to automated transportation introduces mixed-autonomy traffic where human drivers may strategically exploit the risk-averse behavior of Connected and Automated Vehicles (CAVs). While microscopic models capture these dyadic interactions, their aggregate impact on network-level stability remains unquantified due to the scale gap between agent-based and continuum models. This study bridges this analytical divide by integrating a game-theoretic friction term directly into the macroscopic kinematic wave framework. We model the cut-in maneuver as a Stackelberg game, identifying a distinct ``exploitation window'' where human drivers leverage CAV defensiveness to execute aggressive merges. By deriving a closed-form micro-macro bridge, we translate these discrete strategic outcomes into a continuous friction parameter that endogenously modifies the traffic conservation law. Theoretical analysis and numerical simulations confirm that this behavioral asymmetry functions as a deterministic destabilizer, generating perturbation source terms that trigger phantom jams and strictly reduce road capacity. Crucially, we reveal a convex relationship between CAV penetration and system efficiency, identifying a critical instability regime at intermediate penetration rates (approximately 45\%) where the frequency of exploitable interactions is maximized. These findings demonstrate that without socially aware control policies, the defensive nature of early-deployment CAVs may paradoxically degrade traffic flow stability.

\end{abstract}

\section{Introduction}
\label{sec:intro}

The integration of Connected and Automated Vehicles (CAVs) into existing transportation networks promises a transformative shift in traffic efficiency and safety. However, the transition to fully automated systems will not be instantaneous. For the foreseeable future, our roads will be characterized by mixed-autonomy traffic, a hybrid environment where rule-based CAVs must coexist with utility-maximizing Human-Driven Vehicles (HDVs). While much of the existing literature focuses on the cooperative potential of CAVs to smooth traffic flow \cite{talebpour_influence_2016,rios-torres_impact_2017,olia_traffic_2018,zheng_smoothing_2020,li_cooperative_2022,wang_optimal_2022,wu_flow_2022,chen_evidence_2025}, a critical and under-explored dimension is the potential for strategic exploitation. Human drivers, observing the risk-averse and predictable nature of automated systems, may adapt their behavior to ``game'' the system \cite{liu_is_2022}. One example of potential aggressive behavior of HDVs gaming CAVs is cutting in aggressively in front of a CAV with the knowledge that the CAV will yield to avoid collision.

Among the spectrum of competitive driving behaviors, the cut-in maneuver represents the most critical source of friction in mixed-autonomy traffic. Unlike routine lane changes that occur in ample gaps and often facilitate flow homogenization, a cut-in is distinctly characterized by an aggressive entry into a gap smaller than the accepted safety margin, forcing the follower to brake immediately \cite{wang_analysis_2019,gao_discretionary_2022}. In a mixed environment, this maneuver becomes the focal point of strategic exploitation: while a human follower might contest the gap to discourage intrusion, a risk-averse CAV will consistently yield to restore its safety buffer \cite{fu_human-like_2019}. This makes the cut-in not just a local safety hazard, but the primary mechanism by which microscopic behavioral asymmetry degrades macroscopic flow stability \cite{xie_modeling_2019, sultan_modeling_2002}. Therefore, isolating this specific maneuver is essential to quantifying the true cost of mixing defensive CAVs with potentially aggressive HDVs.

Current macroscopic traffic flow models, grounded in the seminal Lighthill-Whitham-Richards (LWR) theory \cite{lighthill1955kinematic,richards1956shock}, typically treat traffic as a homogeneous fluid. While significant advances have been made to incorporate lane-changing effects, most notably through hybrid bounded-acceleration models \cite{laval2006lane} or effective lane-changing intensity parameters \cite{jin2010kinematic}, these approaches predominantly rely on aggregate averages to describe interaction frequency. Consequently, while these models excel at capturing physical congestion phenomena such as shockwave propagation and capacity drops triggered by geometric bottlenecks \cite{cassidy2005increasing}, they struggle to account for behavioral bottlenecks. Even advanced multi-class continuum models \cite{hoogendoorn2000continuum,logghe2008multi,levin2016multiclass}, designed to handle mixed traffic streams, typically differentiate vehicles based on physical constraints (e.g., maximum speed, reaction time) rather than strategic decision-making. They lack the mechanism to describe how individual ``wins'' and ``losses'' in microscopic cut-in games accumulate to degrade macroscopic capacity, leaving a gap in modeling the congestion induced specifically by the heterogeneity of driver strategies.

In a mixed-autonomy setting, the interaction between an aggressive HDV and a defensive CAV is not only a physical event but a game-theoretic one. Foundational studies have long established that lane-changing decisions are best modeled as non-cooperative games, where drivers strategically anticipate the opponent's response rather than reacting solely to physical gaps \cite{kita_merginggiveway_1999,liu_game_2007}. While earlier frameworks focused on mandatory merging \cite{arbis_game_2019}, recent work has extended this to discretionary maneuvers, utilizing concepts like Quantal Response Equilibrium (QRE) to capture the bounded rationality of human decision-making \cite{talebpour_modeling_2015,wang_modeling_2022}. In this context, the decision to cut in is governed by a payoff matrix where the HDV weighs the utility of speed gain against the risk of conflict. However, this negotiation is inherently asymmetric. As noted in recent mixed-traffic studies \cite{yu_row-based_2023}, rational human drivers often adopt a Stackelberg leader role, exploiting the knowledge that a risk-averse CAV (follower in the game) will strictly adhere to programmed safety constraints. Consequently, the ``risk'' term in the HDV's payoff matrix is effectively nullified. This behavioral asymmetry creates a unique form of friction: high-frequency, low-gap cut-ins that force CAVs to brake excessively, generating phantom jams even in the absence of high traffic volume. The phenomenon is confirmed by recent studies identifying cut-ins as a primary catalyst for shockwaves and platoon disruption \cite{shang_cut-ins_2020,lu_modeling_2022}. Comparative assessments further verify that this conservative CAV response significantly degrades local stability compared to pure HDV flows \cite{mullakkal-babu_comparative_2022,fu_human-like_2019}. However, quantifying the macroscopic impact of these microscopic strategic interactions remains an open challenge.

This quantification is particularly difficult because the relevant literature remains bifurcated. Currently, the study of mixed-autonomy traffic generally falls into two disconnected domains:
\begin{itemize}
    \item \textbf{Microscopic Game Theory:} Researchers have extensively modeled vehicle-to-vehicle interactions using differential games and Stackelberg leader-follower frameworks. For instance, Wang et al. \cite{wang_game_2015} demonstrated that lane-changing and car-following are intrinsically coupled non-cooperative games, requiring Model Predictive Control (MPC) to solve for optimal trajectories. While these studies provide rich detail on dyadic interactions, they are computationally prohibitive to scale; simulating thousands of game-theoretic agents to observe network-level effects suffers from the curse of dimensionality.
    \item \textbf{Macroscopic Flow Theory:} Conversely, continuum models efficiently describe network-level dynamics but lack the granularity to capture strategic behavioral shifts. Whether utilizing the lane-changing intensity parameter ($\epsilon$) \cite{jin2010kinematic} or bounded acceleration constraints \cite{laval2006lane}, these frameworks typically model interaction frequency as a function of local traffic density or road geometry, and do not explicitly account for the variable friction caused by mixed-traffic strategic interactions.
\end{itemize}

Notable efforts to bridge this scalability gap have been made using Mean Field Game (MFG) theory. Huang et al. \cite{huang_game-theoretic_2020} introduced a framework connecting microscopic differential games to macroscopic continuum models, allowing for the derivation of optimal velocity control policies in large-scale systems. Similarly, in their stability analysis, Huang et al. \cite{huang_scalable_2020} utilized MFGs to demonstrate that CAVs generally smooth traffic flow and enhance stability.  However, these frameworks primarily address \textit{longitudinal} velocity control and typically conclude that automated agents act as stabilizers. They largely overlook the \textit{lateral} behavioral asymmetry: specifically, the potential for human drivers to ``bully'' risk-averse CAVs during lane changes. Consequently, there remains a distinct lack of analytical frameworks that capture how this specific non-cooperative lateral friction destabilizes the traffic stream.

To address this gap, the primary objective of this study is to bridge the analytical divide between microscopic driver strategy and macroscopic traffic flow stability, with a specific focus on the high-friction cut-in maneuver. Building upon the foundational lane-changing kinematic wave theory established by Jin \cite{jin2010kinematic}, we introduce a novel micro-macro analytical framework designed to quantify how discrete, game-theoretic negotiations between aggressive human drivers and defensive automated systems accumulate to create systemic flow degradation. Central to this framework is the derivation of a closed-form micro-macro bridge function, which calculates the macroscopic lane-changing intensity ($\epsilon$) endogenously based on the instantaneous Stackelberg equilibrium of the microscopic cut-in game, rather than relying on static calibration parameters.

Our contributions are threefold, bridging theoretical modeling with rigorous validation. First, we model the microscopic HDV cut-in maneuver as a Stackelberg game, allowing us to identify a distinct ``exploitation window,'' which is a specific range of gap sizes where rational HDVs will successfully cut off a CAV but would otherwise yield to another HDV. Second, we analytically link this behavior to the continuum scale by deriving a probabilistic aggregation method that maps these type-dependent critical gaps directly to a macroscopic friction parameter, enabling the LWR model to react dynamically to behavioral heterogeneity. Finally, we validate this framework through theoretical proofs and numerical simulations, demonstrating that these strategic interactions strictly reduce capacity and increase shockwave speeds. Crucially, our analysis identifies a convex relationship between capacity and CAV penetration, demonstrating that system performance is minimized during the intermediate stages of adoption.

The remainder of this paper is organized as follows. Section \ref{sec:game_theory} defines the microscopic game-theoretic model. Section \ref{sec:macroscopic_model} outlines the macroscopic flow framework. Section \ref{sec:bridge} details the derivation of the analytical bridge. Section \ref{sec:math_analysis} provides the theoretical proofs of capacity degradation, followed by numerical validation in Section \ref{sec:simulation} and discussion in Section \ref{sec:discussion}. Finally, Section \ref{sec:conclusion} concludes the paper.

\section{Microscopic Game-Theoretic Model for Cut-in Interactions}
\label{sec:game_theory}

To rigorously capture the strategic nature of the cut-in maneuvers discussed in Section \ref{sec:intro}, we establish a microscopic modeling framework that distinguishes between routine driving and adversarial conflict. While standard behavioral models like the Intelligent Driver Model (IDM) and ``Minimizing Overall Braking Induced by Lane Changes'' (MOBIL) effectively govern steady-state car-following and cooperative lane changes, they lack the game-theoretic structure required to model the active negotiation of a cut-in. To address this limitation, we propose a hybrid framework. This architecture integrates the conventional coupled dynamics of IDM and MOBIL with a strategic game-theoretic layer that activates exclusively during Right-of-Way (ROW) conflicts, allowing the HDV agent to transition from reactive driving to strategic utility maximization when an exploitation opportunity arises.

\subsection{Hybrid Control Architecture}
The Ego HDV operates under a switching policy that transitions between reactive lane-keeping and strategic interaction based on the driving context.

\subsubsection{Control Logic}
We formulate the Ego HDV's motion planning as a coupled two-dimensional control problem governed by a hierarchical hybrid architecture. Instead of relying on a monolithic driving policy, the model integrates a reactive layer (IDM and MOBIL) for standard flow stability with a strategic layer (Game Theory) for conflict resolution. The system dynamically arbitrates between these two behavioral modes based on the presence of a Right-of-Way (ROW) competition, ensuring that the vehicle exhibits compliant physics during routine driving while adopting rational utility maximization during critical interactions.

\textbf{Longitudinally}, the commanded acceleration, $u_{x,\text{HDV}}(t)$, switches between a game-theoretic optimal action and a standard car-following law:
\begin{equation}
u_{x,\text{HDV}}(t) =
\begin{cases}
\mathbf{u}_x^*(t), & \text{if } \mathbf{s}(t) \in \mathcal{Z}_{\text{ROW}} \text{ (Strategic Mode)} \\
a_{\text{IDM}}(t), & \text{otherwise (Car-Following Mode)}
\end{cases}
\end{equation}
where $\mathbf{u}_x^*(t)$ is the game-theoretic longitudinal acceleration computed to accomplish the cut-in maneuver, $a_{\text{IDM}}(t)$ is the IDM longitudinal acceleration for regular car-following \cite{treiber2013traffic}, $\mathbf{s}(t)$ is the current system state, and $\mathcal{Z}_{\text{ROW}}$ is the ROW competition zone, which is defined in detail in the following Section \ref{sec:competition_zone}.

\textbf{Laterally}, the control combines continuous lane-keeping with a switching intentional input. The command is:
\begin{equation}
u_{y,\text{HDV}}(t) = -k_\psi \psi_\text{HDV}(t) + u_{y,\text{HDV}}^{\text{intent}}(t)
\end{equation}
where $k_\psi$ is the heading error gain used to dampen rotational oscillations, and $\psi_\text{HDV}(t)$ is the current heading angle relative to the road centerline. The intentional term $u_{y,\text{HDV}}^{\text{intent}}$ follows a three-tiered logic:
\begin{equation}
u_{y,\text{HDV}}^{\text{intent}}(t) = 
\begin{cases}
\mathbf{u}_y^*(t), & \text{if } \mathbf{s}(t) \in \mathcal{Z}_{\text{ROW}} \\
u_{\text{MOBIL}}(t), & \text{if MOBIL criteria met (Regular Lane Change)} \\
u_{\text{keep}}(t), & \text{otherwise (Lane Centering)}
\end{cases}
\end{equation}
where $\mathbf{u}_y^*(t)$ is the game-theoretic steering action for cut-in; $u_{\text{MOBIL}}(t)$ is a constant lateral bias applied towards the target lane, triggered only if both the safety criterion and incentive criterion of the MOBIL model are satisfied \cite{treiber2013traffic}; and $u_{\text{keep}}(t)$ acts as a Proportional-Derivative (PD) lane-centering controller:
$ u_{\text{keep}}(t) = -k_y (y_i(t) - y_{\text{center}}) - k_{vy} v_{y,i}(t) $, 
with $y_{\text{center}}$ the lateral coordinate of the target lane center, and $k_y$ and $k_{vy}$ the proportional and derivative gains penalizing lateral deviation and lateral velocity, respectively.

\subsubsection{ROW Competition Zone}
\label{sec:competition_zone}

The transition from the default car-following mode to the strategic game-theoretic mode is governed by the entry into the ROW Competition Zone, denoted as $\mathcal{Z}_{\text{ROW}}$. We formally define $\mathcal{Z}_{\text{ROW}}$ as the specific region of the state space where the Ego HDV and the TV are in active conflict for the same longitudinal space.

The boundary of this zone is determined by three kinematic metrics: the available following gap ($g_{\text{follow}}$), the relative velocity ($\Delta v$), and the projected Time-to-Collision (TTC). The TTC serves as a critical proxy for urgency, calculated based on the instantaneous closure rate between the TV's front bumper and the Ego HDV's rear bumper:
\begin{equation}
\text{TTC}(t) = 
\begin{cases} 
  \frac{g_{\text{follow}}(t)}{v_{\text{TV}}(t) - v_{\text{HDV}}(t)} & \text{if } v_{\text{TV}}(t) > v_{\text{HDV}}(t) \\
  \infty & \text{otherwise (vehicles separating or static)}
\end{cases}
\end{equation}

A system state $\mathbf{s}(t)$ is considered to be within the competition zone if and only if it satisfies the following concurrent conditions:
\begin{equation}
\mathbf{s}(t) \in \mathcal{Z}_{\text{ROW}} \iff 
\begin{cases} 
  0 < g_{\text{follow}}(t) \le g_{\text{max}}^{\text{comp}} & \text{(Proximity Condition)} \\
  |\Delta v(t)| \le \Delta v_{\text{max}}^{\text{comp}} & \text{(Kinematic Condition)} \\
  \text{TTC}(t) \le \tau_{\text{max}}^{\text{comp}} & \text{(Urgency Condition)}
\end{cases}
\end{equation}
where the parameters define the physical limits of meaningful ROW competition: $g_{\text{max}}^{\text{comp}}$ (e.g., 60 m) represents the maximum contestable gap; beyond this distance, the space is considered ``free,'' and no strategic competition is required. $\Delta v_{\text{max}}^{\text{comp}}$ (e.g., 10 m/s) bounds the speed difference; if the relative speed exceeds this threshold, the vehicles are in a passing regime rather than an interactive merging regime. $\tau_{\text{max}}^{\text{comp}}$ (e.g., 5.0 s) is the safety horizon; infinite or large TTC values imply no immediate collision risk, negating the need for game-theoretic resolution.

\subsection{Game-Theoretic Formulation}
When the vehicle is inside $\mathcal{Z}_{\text{ROW}}$, the interaction is modeled as a two-player, non-cooperative, asymmetric game. We restrict the formulation to a dyadic setting because the cut-in maneuver is fundamentally a pairwise conflict over a specific headway gap; while surrounding traffic provides context, the resolution relies primarily on the lead-follow dynamics between the Ego HDV and the TV. Furthermore, the game is non-cooperative because, unlike connected automated systems that might optimize a joint objective, human drivers act as rational agents maximizing individual utility without binding agreements. Finally, the interaction is inherently asymmetric: the agents occupy distinct strategic roles, with the Ego HDV acting as the ``challenger'' seeking to disrupt the flow and the TV acting as the ``defender'' possessing the initial ROW.

\subsubsection{Players and System State}
The interaction involves two agents:
\begin{itemize}
    \item Player 1: Ego HDV (Strategic Agent). A utility-maximizing agent that optimizes its trajectory to gain the right-of-way.
    \item Player 2: TV (CAV or HDV). The opponent in the target lane, defined by type $k_\text{TV} \in \{\text{CAV}, \text{HDV}\}$.
\end{itemize}

The full system state $\mathbf{s}(t)$ captures the physical kinematics and the environmental context:
\begin{equation}
\mathbf{s}(t) = [ \mathbf{x}_\text{HDV}(t), \mathbf{x}_\text{TV}(t), k_\text{TV}, \mathbf{g}_\text{env}(t) ]^\top
\end{equation}
where the physical state vector for each vehicle $i$ includes position, velocity, and acceleration components:
\begin{equation}
\mathbf{x}_i(t)=\begin{bmatrix} x_i(t) & y_i(t) & v_{x,i}(t) & v_{y,i}(t) & a_{x,i}(t) & \psi_i(t) \end{bmatrix}^\top
\end{equation}
with $x_i(t)$ and $y_i(t)$ denoting the longitudinal and lateral positions, $v_{x,i}(t)$ and $v_{y,i}(t)$ the corresponding velocity components, $a_{x,i}(t)$ representing the longitudinal acceleration, and $\psi_i(t)$ the vehicle's heading angle relative to the road centerline. The environmental vector $\mathbf{g}_\text{env}$ includes specific effective gaps required for the strategic assessment, defined in Table \ref{table:effgaps}.

\begin{table}[h!]
\centering
\renewcommand{\arraystretch}{1.3}
\setlength{\tabcolsep}{4pt}
\caption{Effective Gaps for Strategic Assessment}
\label{table:effgaps}
\small
\begin{tabular}{|l|p{0.35\textwidth}|p{0.35\textwidth}|}
\hline
\textbf{Gap Type} & \textbf{Description} & \textbf{Strategic Role} \\
\hline
Lead Gap ($g_\text{lead}$) & Distance to vehicle ahead of TV. & Represents the reward (available space) for the cut-in. \\
\hline
Follow Gap ($g_\text{follow}$) & Distance from TV to Ego HDV. & Measures the conflict intensity and collision risk. \\
\hline
Escape Gap ($g_\text{escape}$) & Distance to following vehicle in current lane. & Represents the safety net for aborting the merge. \\
\hline
\end{tabular}
\end{table}

\subsubsection{Optimization Problem}
\label{sec:optimization}
The Ego HDV's objective is to solve for the optimal control sequence $(\mathbf{u}_x^*, \mathbf{u}_y^*)$ that maximizes its total expected utility $U_{\text{HDV}}$ over a prediction horizon $H$. This is formulated as a MPC problem:
\begin{equation}
(\mathbf{u}_x^*(t), \mathbf{u}_y^*(t)) = \displaystyle \arg\max_{(\mathbf{u}_x, \mathbf{u}_y) \in \mathcal{A}_\text{HDV}} U_{\text{HDV}} = \displaystyle \arg\max_{(\mathbf{u}_x, \mathbf{u}_y) \in \mathcal{A}_\text{HDV}} \sum_{k=0}^{H} \gamma^k \cdot r_{\text{HDV}}(\mathbf{s}_{t+k}, \mathbf{u}_{t+k}; \pi_\text{TV})
\end{equation}
subject to the state transition dynamics:
\begin{equation}
\mathbf{s}_{t+k+1} = f(\mathbf{s}_{t+k}, \mathbf{u}_{t+k}, \pi_{\text{TV}}(\mathbf{s}_{t+k}))
\end{equation}
where $f(\cdot)$ represents the discrete-time kinematic update of the vehicle states (position, velocity, heading) based on a standard point-mass bicycle model discretized with time step $\Delta t$. $\gamma$ is the discount factor and $\mathcal{A}_\text{HDV}$ represents the admissible action space $\{ (\mathbf{u}_x, \mathbf{u}_y) | \mathbf{u}_x \in [u_x^\text{min}, u_x^\text{max}], \mathbf{u}_y \in [u_y^\text{min}, u_y^\text{max}] \}$, bounded by the vehicle's physical limits.

In this formulation, $\pi_{\text{TV}}$ acts as the forward model of the environment. The Ego HDV does not negotiate with the TV; rather, it predicts the TV's deterministic reaction (e.g., yielding) via state transitions and incorporates that reaction into its own trajectory planning.

The instantaneous reward function $r_{\text{HDV}}$ balances the incentive to merge against safety costs:
\begin{equation}
r_{\text{HDV}}(t) = \underbrace{w_{v} v_{\text{HDV}}(t) f(g_\text{lead})}_{\text{Motivation}} - \underbrace{w_{a} \| \mathbf{u}(t) \|^2}_{\text{Effort}} - \underbrace{w_{c} \Psi_\text{conflict}(t)}_{\text{Safety}}
\end{equation}

Crucially, the conflict penalty $\Psi_{\text{conflict}}(t)$ is determined by the Ego HDV's internal prediction of the opponent's strategy, which relies on a key informational asymmetry regarding the TV's policy $\pi_{\text{TV}}$. We assume the Ego HDV possesses perfect knowledge that a CAV opponent adheres to a strict, deterministic yielding protocol, denoted as $\pi_{\text{safe}}$. Under $\pi_{\text{safe}}$, the target is predicted to decelerate immediately to restore a safe gap, ensuring TTC and resulting in a negligible penalty ($\Psi_{\text{conflict}} \approx 0$) that effectively discounts the risk. In contrast, against an HDV opponent, the Ego agent lacks this assurance and conservatively assumes a non-yielding protocol, denoted as $\pi_{\text{reactive}}$, where the target maintains its velocity and Right-of-Way. This prediction results in a rapidly diminishing TTC, triggering a prohibitive penalty ($\Psi_{\text{conflict}} \gg 0$) that outweighs the utility of the gap and aborts the maneuver.

We formally classify this interaction as a Stackelberg Game, where the Ego HDV acts as the strategic leader and the TV acts as the follower. By assuming that the leader possesses information regarding the follower's reaction function ($\pi_{\text{TV}}$), we reduce the complex game to a hierarchical optimal control problem. This formulation is physically representative of aggressive cut-ins where the initiator forces a reaction from the defender. The formulation is also computationally advantageous, as it avoids the need for iterative Nash equilibrium solutions, allowing for the direct derivation of the closed-form decision boundaries used in the micro-macro bridge.

\section{Macroscopic Flow Model Considering Lane Change Effects}
\label{sec:macroscopic_model}

Having established the microscopic mechanism of strategic exploitation in Section \ref{sec:game_theory}, we now seek to quantify the systemic impact of these interactions on traffic throughput and stability. To achieve this, we adopt the macroscopic framework proposed by Jin in 2010 \cite{jin2010kinematic}, which explicitly incorporates the frictional effects of lane-changing maneuvers into the kinematic wave equation. This framework provides the necessary mathematical structure to translate the discrete, high-frequency cut-in events identified in our game-theoretic model into a continuous, aggregate flow variable.

The core premise of this model is that lane-changing vehicles consume more road capacity than lane-keeping vehicles due to the creation of voids (i.e., safety gaps) and the disruption of following platoons. Jin captures this phenomenon by introducing a dimensionless lane-changing intensity parameter, $\epsilon(x,t)$, which scales the physical density $\rho$ to a higher effective density $\bar{\rho}$:

\begin{equation}
\bar{\rho}(x,t) = \rho(x,t) \left( 1 + \epsilon(x,t) \right)
\end{equation}
where $\rho(x,t)$ is the actual physical density (veh/km) at location $x$ and time $t$, and $\epsilon(x,t) \ge 0$ represents the fractional increase in perceived density caused by turbulence. If $\epsilon=0$, the flow is laminar; if $\epsilon > 0$, the traffic stream behaves as if it were denser than it physically is, leading to speed reduction.

In standard traffic flow theory, equilibrium speed is a function of physical density, $V(\rho)$. However, in the presence of friction, drivers reduce their speed in response to the \textit{effective} density. The modified speed-density relationship becomes:
\begin{equation}
v = V(\bar{\rho}) = V \left( (1+\epsilon)\rho \right)
\end{equation}
Consequently, the macroscopic flow rate $q$ (flux) is depressed. The Fundamental Diagram (FD) describing the flow-density relationship is modified as:
\begin{equation}
q(x,t) = \rho(x,t) \cdot V \left( (1+\epsilon(x,t))\rho(x,t) \right)
\label{eq:modified_flux}
\end{equation}
The evolution of traffic flow is then governed by the inhomogeneous LWR conservation law:
\begin{equation}
\frac{\partial \rho}{\partial t} + \frac{\partial}{\partial x} \left( \rho V((1+\epsilon)\rho) \right) = 0
\end{equation}

The critical innovation in this present paper lies in the derivation of $\epsilon$. While Jin originally formulated $\epsilon$ as a function of macroscopic aggregate variables (e.g., ramp flow ratios or geometric lane drops), we propose that in mixed-autonomy traffic, $\epsilon$ is fundamentally driven by strategic behavioral interactions.

Specifically, the exploitation event modeled in Section \ref{sec:game_theory}, where HDVs cut in front of CAVs, generate a specific magnitude of turbulence that cannot be captured by simple geometric factors. In the following section, we introduce a novel micro-macro bridge that analytically derives $\epsilon(x,t)$ as a function of the game-theoretic cut-in probability and the market penetration rate of CAVs.

\section{Bridging from Microscopic Game to Macroscopic Flow}
\label{sec:bridge}

In this section, we establish the analytical link between the microscopic driver behavior modeled in Section \ref{sec:game_theory} and the macroscopic flow evolution described in Section \ref{sec:macroscopic_model}. While the game-theoretic framework captures the precise mechanics of individual cut-in decisions, integrating these discrete events into a continuum model requires aggregating the cumulative effect of these interactions. We propose a methodology to map the strategic outcomes of the two-player game directly to the macroscopic friction parameter $\epsilon$, enabling the assessment of mixed-autonomy traffic performance without the computational burden of simulating every agent's optimization process.

\subsection{Distilling the Game Model to a Decision Boundary}
The full game-theoretic model for HDV cut-in is a MPC problem that must be solved numerically. While this is accurate for modeling a single HDV's decision-making process, it is computationally infeasible to embed this iterative optimization inside every single HDV agent within a large-scale macroscopic simulation. Therefore, we propose a method to ``distill'' the game model into a decision boundary. The purpose of the macro-micro bridge is to find an analytical link, which requires a simpler, closed-form representation of the game's outcome.

The HDV's strategic decision-making is governed by the instantaneous reward function, $r_{HDV}(t)$, which balances the utility of speed gain against the costs of control effort and safety risks. As defined in Section \ref{sec:optimization}, this function incorporates a positive weighting for speed gain scaled by the leading gap, $f(g_{lead})$, and a severe penalty term, $\Psi_{conflict}(t)$, which acts as a binary switch. This penalty imposes negligible cost during safe interactions but prohibitive cost when the predicted TTC drops below a critical threshold. This structure introduces a distinct non-linearity: a rational HDV attempts a cut-in only when the expected utility of the maneuver exceeds that of staying in the current lane.

Crucially, this utility trade-off is sensitive to the TV's type. When the target is a CAV, the HDV predicts a yielding response, which prevents the TTC from dropping below the threshold even at smaller gaps. Consequently, $\Psi_{conflict}$ remains zero for tighter merges, making the maneuver attractive earlier. Conversely, when the target is an HDV, the prediction of a reactive (or non-yielding) policy triggers $\Psi_{conflict}$ at those same gaps, suppressing the maneuver.

This logic effectively establishes two distinct decision boundaries based on the available gap size in the target lane. We define these boundaries as the type-dependent critical gaps, $g_{crit}^{cut, k_{TV}}$. Consequently, the complex game-theoretic interaction can be robustly proxied by a conditional parameter where a cut-in is attempted if and only if the available gap $g$ satisfies the condition $g > g_{crit}^{cut, k_{TV}}$.

Formally, let the decision to cut in be denoted by the binary variable $D_{\text{cut}} \in \{0, 1\}$. In the full game-theoretic model, this decision is a result of the utility maximization:
\begin{equation}
    D_{\text{cut}}^{\text{Game}}(\mathbf{s}(t)) = \mathbb{I}\left( \max_{\mathbf{u} \in \mathcal{A}_{\text{HDV}}} U_{\text{HDV}}(\mathbf{s}(t), \mathbf{u}; \pi_{\text{TV}}) > U_{\text{stay}} \right)
\end{equation}
where $\mathbb{I}(\cdot)$ is the indicator function and $\pi_{\text{TV}}$ depends on the target type $k_{\text{TV}}$. The term $U_{\text{stay}}$ represents the baseline utility of remaining in the current lane, derived from the standard car-following rewards (maintaining speed and safety) without the additional reward of the lead gap or the penalty of conflict.

In our distilled model, we approximate this complex mapping with a threshold function dependent on the target lane following gap, $g_{follow}$, and the target type:
\begin{equation}
    D_{cut}^{Distilled}(g_{follow}, k_{TV}) = \begin{cases} 
      1 & \text{if } g_{follow} > g_{crit}^{cut, k_{TV}} \\
      0 & \text{otherwise}
   \end{cases}
\end{equation}
The parameters are defined as the gap sizes where the expected utility of the cut-in transitions from negative to positive for each target type:
\begin{equation}
    g_{crit}^{cut, k_{TV}} = \inf \{ g \in \mathbb{R}^+ : U_{cut}(g | k_{TV}) > U_{stay} \}
\end{equation}
Due to the conservative nature of the CAV policy, we derive the inequality $g_{crit}^{cut, CAV} < g_{crit}^{cut, HDV}$. This simplification allows us to replace the computationally expensive on-line optimization with a static, type-dependent check within the macroscopic derivation.

\subsection{Micro-Macro Bridge}

To bridge the microscopic game-theoretic modeling of cut-ins and the macroscopic flow model considering lane change effects, we employ a probabilistic aggregation method. Our objective is to translate the discrete, agent-level cut-in decisions into a continuous expected lane-changing intensity function:
\begin{equation}
\epsilon(\rho, \mathbf{\Theta})
\end{equation}
where $\rho$ is the traffic density and $\mathbf{\Theta}$ represents the vector of system parameters. These parameters include:
\begin{itemize}
    \item Behavioral weights: $w_v$, $w_a$, and $w_c$ governing the HDV utility function.
    \item Traffic composition: The market penetration rate of CAVs, denoted as $m_{\text{C}}$, and the proportion of HDVs $m_{\text{H}} = 1 - m_{\text{C}}$. These determine the frequency of specific HDV-CAV interactions.
    \item Physical constraints: Limits such as $a_{\text{min}}$, $a_{\text{max}}$, desired speed $v_d$, and the lane-change impact duration $t_{\text{LC}}$.
    \item CAV policy: The explicit safety gap $g_{\text{safe}}$ that dictates the automated yielding response.
\end{itemize}

\subsubsection{Probabilistic State Distribution}
First, we require a statistical description of the microscopic states available to drivers. The critical state variable governing the cut-in decision is the available gap $g$ in the target lane.

For the purpose of this analytical derivation, we assume that the headway gaps follow a shifted exponential distribution parameterized by the local density $\rho$. While we select this distribution for its mathematical tractability and reasonable approximation of unsaturated flow, we emphasize that the proposed bridge framework is distribution-agnostic. In real-world applications, any empirically validated distribution $P(g|\rho)$ (e.g., Pearson Type III or log-normal) can be substituted into the subsequent integration steps without altering the fundamental methodology.

Under the exponential assumption, the probability density function for observing a gap of size $g$ is given by:
\begin{equation}
P(g | \rho) = \lambda e^{- \lambda g} \quad \text{for} \, g \ge 0
\end{equation}
where the decay parameter is $\lambda = \frac{\rho}{1-\rho L_{\text{veh}}}$, with $L_{\text{veh}}$ representing the average effective vehicle length.

\subsubsection{Analytical Cut-in Condition}
A rational HDV attempts a cut-in if the expected utility of the maneuver exceeds the utility of maintaining the current lane. From the microscopic game derived in Section \ref{sec:game_theory}, this general decision variable $D_{\text{cut}}$ is defined as:
\begin{equation}
D_{\text{cut}} = 1 \iff U_{\text{cut}} > U_{\text{stay}} \iff g > g_{\text{crit}}^{\text{cut}, k_{\text{TV}}}
\end{equation}
However, for the macroscopic bridge, we must distinguish between standard traffic turbulence and high-friction conflict. We explicitly define a strategic cut-in event, denoted as $\mathcal{E}_{\text{strat}}$, as a maneuver into a gap that is contestable (acceptable via game theory) but smaller than the threshold for a regular, stress-free lane change ($g_{\text{crit}}^{\text{reg}}$).

While $D_{\text{cut}}$ indicates any lane change, $\mathcal{E}_{\text{strat}}$ identifies only those that exploit the target's braking response. The condition for this specific friction-inducing event against a target of type $k_{\text{TV}}$ is:
\begin{equation}
\mathcal{E}_{\text{strat}}(g, k_{\text{TV}}) \iff g_{\text{crit}}^{\text{cut}, k_{\text{TV}}} < g < g_{\text{crit}}^{\text{reg}}
\end{equation}

\subsubsection{Derivation of Lane-Changing Intensity}
To quantify the macroscopic impact, we must first determine the likelihood of each maneuver type. We calculate the probability of a strategic cut-in event ($\mathcal{E}_{\text{strat}}$) by integrating the gap probability density function $P(g|\rho)$ over the contestable region $[\, g_{\text{crit}}^{\text{cut}, k_{\text{TV}}}, \, g_{\text{crit}}^{\text{reg}} \,]$.

Physically, this probability represents the likelihood that a randomly encountered gap falls within the ``aggressive'' bandwidth, which is large enough to satisfy the game-theoretic utility constraint ($U_{\text{cut}} > U_{\text{stay}}$) but small enough to generate friction and require the target to yield. Depending on the target vehicle type ($k_{\text{TV}}$), these probabilities are:
\begin{align}
P_{\text{HDV}\to\text{CAV}}^{\text{cut}}(\rho) &= \int_{g_{\text{crit}}^{\text{cut,CAV}}}^{g_{\text{crit}}^{\text{reg}}} \lambda e^{- \lambda g} dg = e^{- \lambda g_{\text{crit}}^{\text{cut,CAV}}} - e^{- \lambda g_{\text{crit}}^{\text{reg}}} \\
P_{\text{HDV}\to\text{HDV}}^{\text{cut}}(\rho) &= \int_{g_{\text{crit}}^{\text{cut,HDV}}}^{g_{\text{crit}}^{\text{reg}}} \lambda e^{- \lambda g} dg = e^{- \lambda g_{\text{crit}}^{\text{cut,HDV}}} - e^{- \lambda g_{\text{crit}}^{\text{reg}}}
\end{align}

In contrast, we define regular lane changes ($\mathcal{E}_{\text{reg}}$) as discretionary maneuvers into large, uncontested gaps where no braking is required by the follower. The probability of encountering such a cooperative gap is found by integrating the tail of the distribution beyond the regular threshold:
\begin{equation}
P_{\text{HDV}}^{\text{reg}}(\rho) = \int_{g_{\text{crit}}^{\text{reg}}}^\infty P(g | \rho) dg = e^{- \lambda g_{\text{crit}}^{\text{reg}}}
\end{equation}

To determine the macroscopic impact, we adopt Jin's \cite{jin2010kinematic} formulation for lane-changing intensity, $\epsilon \propto \frac{t_{\text{LC}} \cdot v}{L}$. We define the effective impact duration $t_{\text{LC}}$ as a function of the gap size $g$, modeled as a reverse sigmoid. This captures the inverse relationship between gap availability and flow disruption: tighter gaps, while physically executed with faster lateral movement, impose a significantly longer dissipative impact on the following stream due to the induced braking shockwave.
\begin{equation}
    t_{\text{LC}}(g) = t_{\text{min}} + (t_{\text{max}} - t_{\text{min}}) \cdot \frac{1}{1 + e^{k (g - g_{\text{mid}})}}
\end{equation}
where $t_{\text{LC}}^{\text{reg}} \approx t_{\text{min}}$ represents the baseline maneuver time for cooperative lane changes, and $t_{\text{LC}}^{\text{cut}} \approx t_{\text{max}}$ represents the extended congestion impact duration associated with aggressive cut-ins.

Finally, assuming randomly mixed traffic, the total lane-change intensity $\epsilon$ is the sum of the intensities from three distinct behaviors: strategic cut-ins (weighted by the probability of HDV-CAV vs. HDV-HDV pairings) and regular lane changes.
\begin{equation}
\epsilon(\rho | \mathbf{\Theta}) = \epsilon_{\text{HDV}}^{\text{cut}}(\rho) + \epsilon_{\text{HDV}}^{\text{reg}}(\rho) + \epsilon_{\text{CAV}}^{\text{reg}}(\rho)
\end{equation}

The cut-in intensity $\epsilon_{\text{HDV}}^{\text{cut}}$ explicitly captures the asymmetric interaction logic:
\begin{equation}
\label{eq:epsilon_bridge}
\epsilon_{\text{HDV}}^{\text{cut}}(\rho) \approx \frac{v(\rho)}{L} \left[ \underbrace{ (P_{\text{HDV}\to\text{CAV}}^{\text{cut}} \cdot m_{\text{H}} m_{\text{C}}) \cdot t_{\text{LC}}^{\text{cut}} }_{\text{HDV exploits CAV}} + \underbrace{ (P_{\text{HDV}\to\text{HDV}}^{\text{cut}} \cdot m_{\text{H}}^2) \cdot t_{\text{LC}}^{\text{cut}} }_{\text{HDV vs HDV}} \right]
\end{equation}
Substituting the probability expressions:
\begin{equation}
\epsilon_{\text{HDV}}^{\text{cut}}(\rho) \approx \frac{v(\rho) t_{\text{LC}}^{\text{cut}}}{L} \left[ m_{\text{H}} m_{\text{C}} \left(e^{- \lambda g_{\text{crit}}^{\text{cut,CAV}}} - e^{- \lambda g_{\text{crit}}^{\text{reg}}}\right) + m_{\text{H}}^2 \left(e^{- \lambda g_{\text{crit}}^{\text{cut,HDV}}} - e^{- \lambda g_{\text{crit}}^{\text{reg}}}\right) \right]
\end{equation}
The regular lane change intensities are given by:
\begin{align}
\epsilon_{\text{HDV}}^{\text{reg}}(\rho) &\approx m_{\text{H}} e^{- \lambda g_{\text{crit}}^{\text{reg}}} \cdot \frac{v(\rho) t_{\text{LC}}^{\text{reg}}}{L} \\
\epsilon_{\text{CAV}}^{\text{reg}}(\rho) &\approx m_{\text{C}} e^{- \lambda g_{\text{crit,CAV}}^{\text{reg}}} \cdot \frac{v(\rho) t_{\text{LC, CAV}}^{\text{reg}}}{L}
\end{align}
\noindent where $g_{\text{crit,CAV}}^{\text{reg}}$ and $t_{\text{LC, CAV}}^{\text{reg}}$ denote the programmed safety gap and impact duration for CAVs, respectively. We generally assume $g_{\text{crit,CAV}}^{\text{reg}} \ge g_{\text{crit}}^{\text{reg}}$, reflecting the conservative calibration of CAVs even during non-adversarial maneuvers. Additionally, while the formulation allows for distinct impact durations, we approximate $t_{\text{LC, CAV}}^{\text{reg}} \approx t_{\text{LC}}^{\text{reg}}$, as voluntary CAV lane changes are performed only when safety criteria are strictly met, generating turbulence profiles similar to standard non-aggressive maneuvers. This analytical function $\epsilon(\rho)$ is then directly substituted into the macroscopic conservation law derived in Section \ref{sec:macroscopic_model}, closing the loop between the microscopic game and the macroscopic flow evolution.

\section{Analysis of Cut-in Impacts on Traffic Flow}
\label{sec:math_analysis}

In this section, we utilize the proposed micro-macro framework to mathematically analyze the systemic impact of strategic cut-ins on macroscopic traffic dynamics. Having established the analytical link between the microscopic Stackelberg game and the macroscopic lane-changing intensity $\epsilon$, we now proceed to rigorously quantify the resulting alterations to traffic flow. We derive the analytical conditions under which these game-theoretic interactions inevitably lead to system degradation. Specifically, we present formal proofs demonstrating that the existence of a microscopic ``exploitation window'' ($g_{\text{crit}}$) mathematically guarantees a reduction in macroscopic capacity, a decrease in equilibrium speeds, and the premature onset of flow instability. These analytical results serve to validate that the capacity drops and phantom jams observed in mixed-autonomy traffic are intrinsic mathematical consequences of the behavioral friction, rather than numerical artifacts.

\subsection{Assumptions}
We begin by stating the core assumptions that define our theoretical model, upon which our proofs are based.

\begin{assumption}[HDV Behavior]
The HDV is a strategic agent whose actions are chosen to maximize a utility function. The instantaneous reward, $r_{\text{HDV}}(t)$, includes a positive weight for speed gain, denoted $w_v > 0$.
\end{assumption}

\begin{assumption}[CAV Behavior]
The CAV is a rule-based agent following a fixed, predictable, and safety-oriented policy, $\pi_{\text{safe}}$. This policy is yielding, meaning the CAV will decelerate to restore a safe following distance when its safety threshold is breached by another vehicle.
\end{assumption}

\begin{assumption}[Traffic Environment]
The traffic environment is characterized by the following properties:
\begin{enumerate}
    \item A mixed composition of HDVs and CAVs with non-zero market shares, $m_{\text{H}} > 0$ and $m_{\text{C}} > 0$ for HDVs and CAVs, respectively.
    \item The existence of type-dependent critical gaps where the ``exploitative'' gap for CAVs is strictly smaller than the regular gap: $0 < g_{\text{crit}}^{\text{cut,CAV}} < g_{\text{crit}}^{\text{reg}}$.
    \item The probability distribution of available gaps, $g$, for a given traffic density, $\rho$, follows an exponential distribution, $P(g|\rho)=\lambda e^{-\lambda g}$, where $\lambda = \rho / (1-\rho L_{\text{veh}}) > 0$ for $\rho > 0$.
\end{enumerate}
\end{assumption}

\begin{assumption}[Macroscopic Flow Dynamics, i.e., Jin's \cite{jin2010kinematic} framework]
\label{assum:flow_dynamics}
The macroscopic traffic flow, $q$, is governed by a model incorporating a lane-changing intensity parameter, $\epsilon$, given by $q=\rho V((1+\epsilon)\rho)$. The velocity function, $V(\cdot)$, is a strictly decreasing function of its argument, the effective density $\bar{\rho} = (1+\epsilon)\rho$.
\end{assumption}

\subsection{Propositions on Microscopic Strategic Interactions}
In this subsection, we establish the fundamental properties of the game-theoretic interactions. We characterize the structure of the optimal decision policy, prove the existence of exploitative behavior, and decompose its contribution to macroscopic friction.

\subsubsection*{Characterization of the Decision Boundary}
First, it is necessary to verify that the complex utility maximization problem defined in the game reduces to a physically interpretable decision rule. While the utility function considers multiple variables (safety, speed, jerk), rational driving behavior typically manifests as a monotonic response to gap availability. Proposition \ref{prop:threshold_optimality} confirms that our utility formulation collapses into a clean threshold logic.

\begin{proposition}[Optimality of Threshold Policy]
\label{prop:threshold_optimality}
Under the instantaneous reward function defined in Eq. (18), the optimal cut-in decision is governed by a unique threshold policy. Specifically, the utility difference $\Delta U(g) = U_{\text{cut}}(g) - U_{\text{stay}}$ is a strictly increasing function of the available gap $g$, ensuring that a unique critical gap $g_{\text{crit}}$ exists such that a cut-in is optimal if and only if $g > g_{\text{crit}}$.
\end{proposition}

\begin{proof}
The instantaneous reward for a cut-in is given by:
\begin{equation}
    r_{\text{HDV}}(g) = w_{v} v_{\text{HDV}} f(g) - w_{a} \|\mathbf{u}_{\text{HDV}}\|^2 - w_{c} \Psi_{\text{conflict}}(g)
\end{equation}
We analyze the sensitivity of the reward components with respect to the gap size $g$:
\begin{itemize}
    \item Speed Reward: The term $f(g)$ is defined as a sigmoid function of the gap, which is strictly monotonic increasing: $\frac{\partial f}{\partial g} > 0$.
    \item Control Cost: A larger gap allows for a smoother merging trajectory with lower required acceleration magnitudes ($\mathbf{u}_{x}, \mathbf{u}_{y}$), implying the control cost is non-increasing with respect to gap size: $\frac{\partial \|\mathbf{u}\|^2}{\partial g} \le 0$.
    \item Safety Penalty: The binary penalty $\Psi_{\text{conflict}}(g)$ is a step function that transitions from 1 to 0 as $g$ increases beyond the safety threshold. Thus, it is non-increasing.
\end{itemize}
Combining these terms, the total derivative of the utility with respect to the gap is strictly positive:
\begin{equation}
    \frac{\partial r_{\text{HDV}}}{\partial g} > 0
\end{equation}
Furthermore, as $g \to 0$, the safety penalty dominates ($r_{\text{HDV}} < 0$). As $g \to \infty$, the speed reward dominates ($r_{\text{HDV}} > 0$). By the Intermediate Value Theorem and the property of strict monotonicity, there exists a unique root $g_{\text{crit}}$ such that $r_{\text{HDV}}(g_{\text{crit}}) = 0$. Consequently, $r_{\text{HDV}}(g) > 0$ if and only if $g > g_{\text{crit}}$.
\end{proof}

\noindent \textit{Remark:} This proposition is computationally significant. It implies that for the purpose of macroscopic modeling, we do not need to solve the full optimization problem at every time step. Instead, we can pre-calculate the critical gap $g_{\text{crit}}$ as a static behavioral parameter for each driver type, simplifying the aggregation process.

\subsubsection*{Verification of Exploitative Behavior}
Having established the decision boundary, we next address the core hypothesis of this study: that mixed-autonomy traffic generates a specific subset of aggressive maneuvers that would not occur in pure HDV traffic. Proposition \ref{prop:exploit_prob} mathematically validates the existence of this ``exploitation window.''

\begin{proposition}[Existence of Exploitative Cut-in Probability]
\label{prop:exploit_prob}
Within the framework of Assumptions 1-4, the probability of an HDV performing a strategic cut-in specifically by exploiting the CAV's yielding policy is strictly positive.
\end{proposition}

\begin{proof}
The game-theoretic cut-in is defined by the interval $(g_{\text{crit}}^{\text{cut,CAV}}, g_{\text{crit}}^{\text{reg}})$. This represents maneuvers that would be rejected under regular lane-changing logic (where $g < g_{\text{crit}}^{\text{reg}}$) but are accepted due to the strategic game against a CAV. Integrating the gap distribution over this region yields:
\begin{equation}
    P_{\text{exploit}}(\rho) = \int_{g_{\text{crit}}^{\text{cut,CAV}}}^{g_{\text{crit}}^{\text{reg}}} \lambda e^{-\lambda g} \,dg = e^{-\lambda g_{\text{crit}}^{\text{cut,CAV}}} - e^{-\lambda g_{\text{crit}}^{\text{reg}}}
\end{equation}
By Assumption 3, the arrival rate $\lambda > 0$. By Assumption 2, the game theoretic solution yields a critical gap strictly smaller than the regular gap ($g_{\text{crit}}^{\text{cut,CAV}} < g_{\text{crit}}^{\text{reg}}$). Since the exponential function $e^{-x}$ is strictly decreasing, the term $e^{-\lambda g_{\text{crit}}^{\text{cut,CAV}}}$ is strictly larger than $e^{-\lambda g_{\text{crit}}^{\text{reg}}}$, rendering the probability $P_{\text{exploit}}(\rho)$ strictly positive.
\end{proof}

\noindent \textit{Remark:} This result confirms that the phenomenon of ``phantom jams'' in mixed traffic is not merely a stochastic occurrence but a structural inevitability. As long as CAVs are programmed to be more conservative than the average human driver ($g_{\text{crit,CAV}} > g_{\text{crit,HDV}}$), a non-zero probability of exploitation $P_{\text{exploit}}$ is guaranteed to exist.

\subsubsection*{Quantification of Induced Friction}
Finally, we decompose the macroscopic impact of these microscopic games. Proposition \ref{prop:marginal_intensity} separates the lane-changing intensity into distinct components, isolating the specific contribution of strategic gaming.

\begin{proposition}[Decomposition of Marginal Intensity]
\label{prop:marginal_intensity}
The increase in lane-changing intensity due to gaming, denoted as the marginal intensity $\Delta \epsilon_{\text{game}}$, is the product of the Aggression Factor (the expanded probability window) and the Opportunity Factor (the frequency of HDV-CAV pairings).
\end{proposition}

\begin{proof}
Let $\epsilon_{\text{baseline}}$ represent the intensity if all vehicles followed the regular lane-changing logic (using $g_{\text{crit}}^{\text{reg}}$), and $\epsilon_{\text{total}}$ represent the intensity with the game model. The difference is derived from the cut-in component:
\begin{equation}
\Delta \epsilon_{\text{game}}(\rho) = \epsilon_{\text{total}}(\rho) - \epsilon_{\text{baseline}}(\rho)
\end{equation}
Substituting the probability terms derived in the Micro-Macro bridge:
\begin{equation}
\Delta \epsilon_{\text{game}}(\rho) = \underbrace{ (m_{\text{H}} m_{\text{C}}) }_{\text{Opportunity Factor}} \cdot \underbrace{ \left( e^{-\lambda g_{\text{crit}}^{\text{cut,CAV}}} - e^{-\lambda g_{\text{crit}}^{\text{reg}}} \right) }_{\text{Aggression Factor}} \cdot \frac{t_{\text{LC}} v(\rho)}{L}
\end{equation}
Here, the Opportunity Factor ($m_{\text{H}} m_{\text{C}}$) quantifies the probability of an HDV trailing a CAV, which creates the necessary condition for the game. The Aggression Factor corresponds to $P_{\text{exploit}}$, representing the additional volume of lane changes generated because the HDV is willing to accept smaller gaps ($g_{\text{crit}}^{\text{cut,CAV}}$) than usual. Since $m_{\text{H}}, m_{\text{C}} > 0$ and $P_{\text{exploit}} > 0$ (from Proposition \ref{prop:exploit_prob}), the marginal intensity is strictly positive: $\Delta \epsilon_{\text{game}} > 0$.
\end{proof}

\noindent \textit{Remark:} This decomposition highlights the non-linear relationship between penetration rate and flow stability. The friction is maximized not when CAVs are dominant, but when the product $m_H m_C$ is maximized at a 50\% penetration rate, and the opportunity for heterogeneous interaction is highest.

\subsection{Degradation of Macroscopic Efficiency}
We quantify the systemic impact of strategic cut-ins on the steady-state performance metrics of the traffic network. By integrating the game-theoretic friction term into the macroscopic conservation law, we analyze how the equilibrium state of the traffic stream is altered by the presence of exploitative actors. We demonstrate that the microscopic behavioral asymmetry (i.e., HDV's exploitation of CAV defensiveness) does not merely create local turbulence, but translates directly into a strict contraction of the FD. The following theorem rigorously establishes the degradation of three critical efficiency metrics: maximum flow capacity, equilibrium speed, and total system travel time.

\begin{theorem}[Macroscopic Efficiency Degradation]
\label{thm:efficiency}
The presence of game-theoretic cut-ins, quantified by the strictly positive marginal intensity $\Delta \epsilon_{\text{game}}$, strictly degrades the macroscopic efficiency of the traffic stream. Specifically, for any traffic state with density $\rho > 0$, the following conditions hold strictly:
\begin{enumerate}
    \item Capacity Reduction: The fundamental flow-density curve is pointwise suppressed, leading to a strictly lower maximum capacity $q_{\max}$.
    \item Speed Reduction: The equilibrium traffic speed $v(\rho)$ is strictly reduced compared to the baseline.
    \item Travel Time Increase: The total system travel time $T$ is strictly increased.
\end{enumerate}
\end{theorem}

\begin{proof}
We prove these properties sequentially by analyzing the modified conservation law $q(\rho) = \rho V(\rho (1 + \epsilon_{\text{total}}))$.

\textit{Part 1: Capacity and Flow Suppression.}
Let $q_{\text{base}}(\rho)$ be the flow in the absence of gaming. We compare the effective densities governing the two states:
\begin{equation}
    \bar{\rho}_{\text{total}} = \rho (1 + \epsilon_{\text{base}} + \Delta \epsilon_{\text{game}}) \quad \text{vs.} \quad \bar{\rho}_{\text{base}} = \rho (1 + \epsilon_{\text{base}})
\end{equation}
By Proposition \ref{prop:marginal_intensity}, the marginal intensity $\Delta \epsilon_{\text{game}}$ is strictly positive. Consequently, the effective density in the game scenario is strictly higher: $\bar{\rho}_{\text{total}} > \bar{\rho}_{\text{base}}$.
Since the speed function $V(\cdot)$ is strictly decreasing (Assumption \ref{assum:flow_dynamics}), the equilibrium velocity must decrease:
\begin{equation}
    V(\bar{\rho}_{\text{total}}) < V(\bar{\rho}_{\text{base}})
\end{equation}
Multiplying by the physical density $\rho$, we find that the flow curve is pointwise suppressed for all $\rho > 0$:
\begin{equation}
    q_{\text{total}}(\rho) < q_{\text{base}}(\rho)
\end{equation}
Since the entire flow function is suppressed, its maximum value (i.e., capacity) is strictly reduced: $q_{\max}^{\text{total}} < q_{\max}^{\text{base}}$. This confirms that capacity loss is not random, but is a deterministic function of the HDV's aggressiveness (lowering $g_{\text{crit}}^{\text{cut}}$) and the CAV market penetration ($m_C$).

\textit{Part 2: Speed Reduction.}
The average speed is given by $v(\rho) = V(\bar{\rho})$. We decompose the total effective density into the baseline component and the game-induced turbulence:
\begin{equation}
    \bar{\rho}_{\text{total}} = \bar{\rho}_{\text{base}} + \underbrace{\rho \cdot \Delta \epsilon_{\text{game}}(\rho)}_{\text{Virtual Density}}
\end{equation}
The term $\rho \cdot \Delta \epsilon_{\text{game}}$ represents the physical density of ``virtual'' vehicles added by the negotiation friction. Since this term is strictly positive, the system behaves as if it is more crowded than it physically is. By the monotonicity of $V$, it follows immediately that $v_{\text{total}}(\rho) < v_{\text{base}}(\rho)$.

\textit{Part 3: Travel Time.}
The total travel time over segment $L$ is defined as $T(\rho) = L / v(\rho)$. From Part 2, we can express the game-theoretic speed as the baseline speed minus a strictly positive velocity loss term $\delta_v(\Delta \epsilon_{\text{game}})$:
\begin{equation}
    v_{\text{total}}(\rho) = v_{\text{base}}(\rho) - \delta_v
\end{equation}
Substituting this into the travel time equation:
\begin{equation}
    T_{\text{total}}(\rho) = \frac{L}{v_{\text{base}}(\rho) - \delta_v} > \frac{L}{v_{\text{base}}(\rho)} = T_{\text{base}}(\rho)
\end{equation}
The inequality holds strictly because the denominator is reduced by the cumulative friction of exploitative cut-ins, mathematically guaranteeing an increase in system travel time.
\end{proof}

\noindent \textit{Remark:} This theorem establishes that the observed capacity degradation is not a stochastic artifact, but a deterministic structural property of mixed-autonomy traffic. The marginal intensity $\Delta \epsilon_{\text{game}}$ functions effectively as a ``frictional multiplier'' that inflates the physical density $\rho$ into a higher effective density $\bar{\rho}$. Physically, this implies that turbulence generated by the negotiations consumes road capacity just like adding extra vehicles into the stream, making it perceived by the driver to be significantly more congested than the physical vehicle count would suggest.

\subsection{Instability and Congestion Dynamics}
Having established the steady-state efficiency losses in Theorem \ref{thm:efficiency}, we now address the dynamic stability of the traffic stream. While capacity reduction describes the limit of the system, it does not explain the mechanism of breakdown. In this subsection, we rigorously analyze how the spatial inhomogeneity of game-theoretic interactions, specifically the transition into the weaving zone, acts as an active destabilizing force. We begin by defining the perturbation source term derived from the linearized kinematic wave model \cite{jin2010kinematic}, which identifies the specific mathematical pathway through which microscopic behavioral friction triggers macroscopic density nucleation.

\begin{definition}[Perturbation, $\mathcal{P}$]
The perturbation $\mathcal{P}$ is the source term in the linearized density transport equation $\rho_t + \lambda \rho_x = \mathcal{P}$. Physically, it represents the rate of density accumulation caused strictly by lane-changing gradients, defined as:
\begin{equation}
    \mathcal{P} = - \rho^2 V' \frac{\partial \epsilon}{\partial x}
\end{equation}
where $\rho$ is the traffic density, $V'$ is the derivative of the speed-density function ($V' < 0$), and $\frac{\partial \epsilon}{\partial x}$ is the spatial gradient of the lane-changing intensity \cite{jin2010kinematic}. In standard LWR traffic flow, this term is zero ($\mathcal{P} = 0$). A non-zero perturbation $\mathcal{P} > 0$ implies an external forcing term acts on the traffic stream to artificially accelerate density accumulation.
\end{definition}

\begin{theorem}[Destabilization of Traffic Flow]
\label{thm:instability}
The spatial gradients of game-induced friction act as destabilizing forces within the kinematic wave model. This results in three strictly negative stability outcomes:
\begin{enumerate}
    \item Forced Compression: The gradient of $\epsilon$ generates a strictly positive density perturbation $\mathcal{P} > 0$. This forces the temporal rate of density change to be positive ($\frac{\partial \rho}{\partial t} > 0$) even in a geometrically homogeneous road section where the convective gradient is zero.
    \item Premature Congestion: The critical density $\rho_{\text{crit}}$ required to trigger traffic jams is strictly reduced.
    \item Aggravated Shockwaves: Once congestion forms, the backward propagation speed $|w|$ of the shockwave is strictly increased.
\end{enumerate}
\end{theorem}

\begin{proof}
We prove these properties based on the eigenvalues and Riemann solutions of the system.

\textit{Part 1: Forced Compression (Perturbation).}
We analyze the linearized conservation law derived in \cite{jin2010kinematic}. By isolating the transport equation for density, we obtain:
\begin{equation}
    \underbrace{\frac{\partial \rho}{\partial t} + \lambda_1 \frac{\partial \rho}{\partial x}}_{\text{Convection Rate}} = \underbrace{- \rho^2 V' \frac{\partial \epsilon}{\partial x}}_{\text{Perturbation } (\mathcal{P})}
\end{equation}
To determine the stability impact, we analyze the sign of the perturbation term $\mathcal{P}$ at the onset of the weaving section:
1) The squared density $\rho^2$ is strictly positive for any existing flow.
2) The speed gradient $V'$ is strictly negative in stable traffic ($V' < 0$).
3) The friction gradient at the entrance of the game zone intensity increases, so $\frac{\partial \epsilon}{\partial x} > 0$.

Combining these factors:
\begin{equation}
    \mathcal{P} = (-1) \cdot \underbrace{(\rho^2)}_{+} \cdot \underbrace{(V')}_{-} \cdot \underbrace{\left(\frac{\partial \epsilon}{\partial x}\right)}_{+} > 0
\end{equation}
Consequently, even if the incoming traffic is spatially uniform ($\frac{\partial \rho}{\partial x} = 0$), the equation forces $\frac{\partial \rho}{\partial t} > 0$. This proves that the game-theoretic friction acts as a source term, creating a localized compression zone solely due to behavioral heterogeneity.

\textit{Part 2: Premature Congestion (Critical Density).}
The critical density is defined as the unique state where the characteristic wave speed $\lambda_1$ vanishes, marking the transition from free-flow to congestion. Using the eigenvalue derived in \cite{jin2010kinematic}:
\begin{equation}
    \lambda_1(\rho) = V(\bar{\rho}) + \rho(1+\epsilon)V'(\bar{\rho})
\end{equation}
where $\bar{\rho} = (1+\epsilon)\rho$. We compare the equilibrium states:
\begin{itemize}
    \item Baseline ($\epsilon=0$): The zero-crossing occurs at $\rho_{\text{base}}$ such that $V(\rho_{\text{base}}) + \rho_{\text{base}} V'(\rho_{\text{base}}) = 0$.
    \item Game Scenario ($\epsilon > 0$): The negative slope term $V'$ is amplified by $(1+\epsilon)$. To maintain the balance $\lambda_1 = 0$, the equilibrium must shift to a lower density.
\end{itemize}
Mathematically, the relationship is derived as:
\begin{equation}
    (1+\epsilon)\rho_{\text{total}} = \rho_{\text{base}} \implies \rho_{\text{total}} = \frac{\rho_{\text{base}}}{1+\epsilon}
\end{equation}
Since $\epsilon > 0$, it follows strictly that $\rho_{\text{total}} < \rho_{\text{base}}$. Physically, this implies that the traffic stream becomes unstable at densities that would otherwise sustain stable free-flow conditions.

\textit{Part 3: Aggravated Shockwaves.}
The shockwave speed $w$ is governed by the Rankine-Hugoniot condition across the interface:
\begin{equation}
    w = \frac{Q_{\text{supply}} - Q_{\text{demand}}}{\rho_{\text{supply}} - \rho_{\text{demand}}}
\end{equation}
According to the Riemann problem solution for this system (specifically Types 4 and 9 in \cite{jin2010kinematic}), the flow through the weaving section is limited by the downstream supply function:
\begin{equation}
    Q_{\text{supply}} = \frac{1}{1+\epsilon} Q_{\max}^{\text{base}}
\end{equation}
Consider a fixed upstream demand $Q_{\text{demand}}$ that creates congestion. Since $\epsilon > 0$, the supply is strictly reduced ($Q_{\text{supply}} < Q_{\max}^{\text{base}}$), making the flux difference $\Delta Q$ more negative than in the baseline case. Consequently, the magnitude of the backward wave speed strictly increases:
\begin{equation}
    |w_{\text{total}}| = \frac{|\Delta Q_{\text{total}}|}{\Delta \rho} > \frac{|\Delta Q_{\text{base}}|}{\Delta \rho} = |w_{\text{base}}|
\end{equation}
This proves that the queue propagates upstream at a higher velocity due to the reduced discharge capacity caused by game-theoretic friction.
\end{proof}

\noindent \textit{Remark:} This theorem provides the mathematical justification for the phantom jam. Collectively, the three parts of the proof demonstrate that the behavioral gradient $\frac{\partial \epsilon}{\partial x}$ functions as a virtual bottleneck that is dynamically more severe than a static geometric restriction. It actively pumps density into the system via the source term $\mathcal{P}$ (Part 1), lowers the system's immunity to breakdown by reducing $\rho_{\text{crit}}$ (Part 2), and ensures that once congestion forms, the resulting queue propagates upstream more rapidly than in baseline traffic (Part 3). This confirms that instability in mixed-autonomy traffic is a deterministic structural consequence of the friction gradient and not just a stochastic nuisance.

\section{Numerical Validation and Simulation}
\label{sec:simulation}

We empirically validate the theorems presented in Section \ref{sec:math_analysis} by implementing the coupled micro-macro model in a controlled numerical environment. This experimental design allows us to isolate the specific mechanisms driving flow degradation, free from the external noise associated with site-specific calibration. Our primary objective is to verify the macroscopic efficiency degradation identified in Theorem \ref{thm:efficiency} (capacity reduction) and the destabilization dynamics identified in Theorem \ref{thm:instability} (perturbation generation, premature congestion, and aggravated shockwaves). These simulations serve to demonstrate that the observed capacity loss and flow instability are necessary mathematical consequences of the game-theoretic interactions defined in Section \ref{sec:game_theory}.

\subsection{Experimental Design}

\subsubsection{Simulation Environment}
We implement a discrete numerical solver for the macroscopic conservation law derived in Section \ref{sec:macroscopic_model} (i.e., Jin's framework \cite{jin2010kinematic}). The traffic flow evolution is governed by the inhomogeneous kinematic wave equation:
\begin{equation}
    \frac{\partial \rho}{\partial t} + \frac{\partial}{\partial x} \left( \rho V((1+\epsilon(x, \rho))\rho) \right) = 0
\end{equation}
To solve this partial differential equation (PDE), we utilize a first-order Godunov scheme, which is well-suited for capturing the shockwaves and discontinuities inherent in traffic breakdown. The road segment is discretized into cells of length $\Delta x = 50$ m over a total length of $L_{road} = 2$ km. The simulation time step $\Delta t$ is selected to satisfy the Courant-Friedrichs-Lewy (CFL) condition ($v_f \Delta t \le \Delta x$) to ensure numerical stability.

\begin{figure}[h!]
    \centering
    \includegraphics[width=\linewidth]{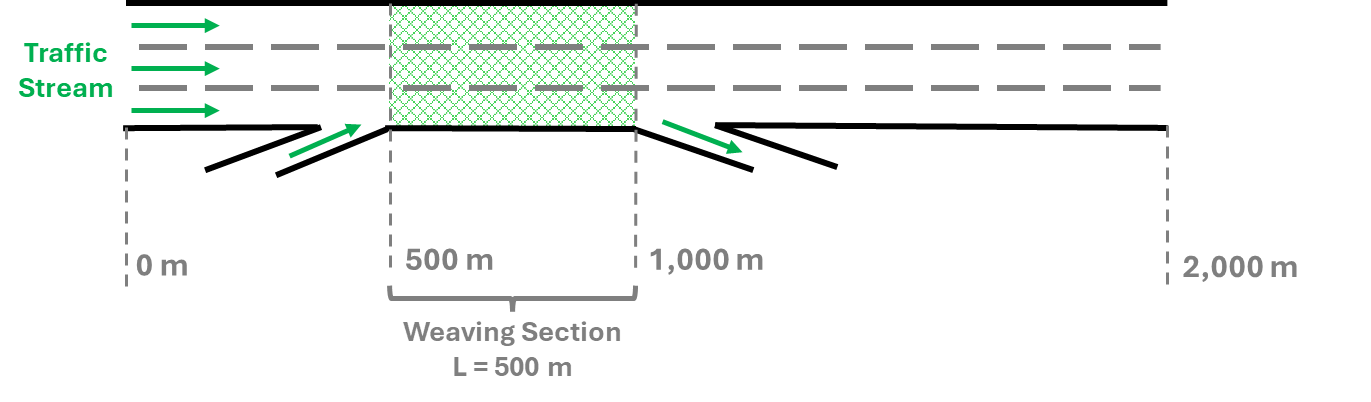}
    \caption{Schematic of the simulation environment. The weaving section, where lane-changing interactions occur ($\epsilon > 0$), is located between 500 m and 1000 m along the 2 km road segment.}
    \label{fig:weaving_setup}
\end{figure}

The road geometry, illustrated in Figure \ref{fig:weaving_setup}, consists of a simplified weaving section located between $x=500$ m and $x=1000$ m. Outside this region, the lane-changing intensity is zero ($\epsilon = 0$). Inside this region, $\epsilon$ is determined dynamically by the micro-macro bridge proposed in the previous section.

\subsubsection{Implementation of the Micro-Macro Bridge}
A critical feature of this simulation is that the lane-changing intensity parameter $\epsilon$ is not a fixed exogenous input. Instead, it is calculated endogenously at every time step $t$ for every spatial cell $j$ based on the local traffic density. To achieve this, the road length $L_{road}$ is discretized into $N$ homogeneous segments (i.e., cells) of length $\Delta x$, and the simulation advances in discrete time steps $\Delta t$.

We implement the analytical bridge function derived in Eq. \eqref{eq:epsilon_bridge} directly into the numerical flux loop. At each time step $t$, the solver performs the following updates for every cell $j$ located within the weaving section:
\begin{enumerate}
    \item Reads the local density $\rho_j^t$ currently stored in the grid cell.
    \item Calculates the exploitative cut-in probability $P_{\text{exploit}}(\rho_j^t)$ based on the game-theoretic threshold policy defined in Proposition \ref{prop:exploit_prob}.
    \item Updates the local friction parameter $\epsilon_j$ by combining the background friction from regular lane changes with the additional friction from strategic conflicts:
    \begin{equation}
        \epsilon_j(\rho_j^t) = \epsilon_{reg}(\rho_j^t) + \alpha \cdot P_{\text{exploit}}(\rho_j^t) \cdot (m_H m_C)
    \end{equation}
    where $\alpha$ is a scaling factor representing the specific impact of exploitative maneuvers on the traffic stream.
    \item Computes the numerical fluxes at the cell boundaries using the modified FD $Q(\rho) = \rho V((1+\epsilon_j)\rho)$ and updates the cell density according to the conservation law:
    \begin{equation}
        \rho_j^{t+1} = \rho_j^t - \frac{\Delta t}{\Delta x} \left( \Phi_{j+1/2} - \Phi_{j-1/2} \right)
    \end{equation}
\end{enumerate}
This coupling ensures that the macroscopic capacity degradation observed in the results is a direct mathematical consequence of the microscopic strategic interactions, without requiring the computational overhead of agent-based simulation.

\subsubsection{Behavioral Scenarios}
To assess the sensitivity of traffic flow to game-theoretic variables, we define two primary behavioral profiles:
\begin{itemize}
    \item Baseline Scenario (non-strategic): Represents a traditional traffic stream where lane changes are discretionary or cooperative. The critical gap for lane changing is fixed at the standard safety threshold ($g_{crit} = g_{crit}^{reg}$), effectively setting the exploitative probability ($P_{\text{exploit}}$) to zero.
    \item Game-Theoretic Scenario (strategic): Represents a mixed-autonomy traffic stream characterized by strategic exploitation. Consistent with the equilibrium derived in Section \ref{sec:game_theory}, human drivers identify CAVs as risk-averse agents and adopt an exploitation strategy. We model this behavior by setting the critical gap for cutting in front of a CAV ($g_{crit}^{cut}$) to be significantly lower than the regular safety gap ($g_{crit}^{reg}$). We vary the CAV market penetration rate ($m_C$) to observe how the frequency of these aggressive interactions, governed by the term $P_{\text{exploit}} \cdot m_H m_C$, impacts flow stability.
\end{itemize}
The specific numerical values for the road geometry, kinematic wave model, and game-theoretic interaction thresholds used in these experiments are summarized in Table \ref{tab:sim_params}.

\begin{table}[h!]
\centering
\caption{Simulation Parameters}
\label{tab:sim_params}
\begin{tabular}{|l|l|c|l|}
\hline
\textbf{Category} & \textbf{Parameter} & \textbf{Symbol} & \textbf{Value} \\
\hline
\textit{Road and Grid} & Road Length & $L_{road}$ & 2000 m \\
 & Cell Size & $\Delta x$ & 50 m \\
 & Time Step & $\Delta t$ & 1.0 s \\
 & Weaving Section & $x_{weave}$ & 500 -- 1000 m \\
\hline
\textit{Traffic Physics} & Free Flow Speed & $v_f$ & 30 m/s (108 km/h) \\
 & Jam Density & $\rho_j$ & 0.12 veh/m (120 veh/km) \\
 & Wave Speed & $C_j$ & 6.0 m/s \\
 & Vehicle Length & $L_{veh}$ & 5.0 m \\
\hline
\textit{Game and Interaction} & Regular Lane Change Impact Duration & $t_{LC}^{reg}$ & 2.0 s \\
 & Aggressive Lane Change Impact Duration & $t_{LC}^{cut}$ & 20.0 s \\
 & Regular Critical Gap & $g_{crit}^{reg}$ & 30.0 m \\
 & Aggressive Critical Gap & $g_{crit}^{cut}$ & 10.0 m \\
\hline
\end{tabular}
\end{table}

\subsection{Simulation Results}
Following the experimental design outlined above, we present the numerical findings in four stages. First, we validate the microscopic behavioral assumptions (Section \ref{sec:micro_validation}) to confirm that the dual critical gap parameters used in the strategic cut-in scenario are consistent with rational utility maximization. Second, we verify the physical mechanism of friction-induced compression (Section \ref{sec:density_nucleation}). Third, we quantify the impact on macroscopic capacity and the FD (Section \ref{sec:capacity_drop}). Finally, we analyze the spatiotemporal instability and shockwave dynamics (Section \ref{sec:shockwaves}).

\subsubsection{Validation of Microscopic Decision Boundaries}
\label{sec:micro_validation}

Before investigating the macroscopic flow dynamics, we first validated the structural assumptions of the micro-macro bridge (Section \ref{sec:bridge}) using the game-theoretic trajectory solver described in Appendix A. The objective was to verify whether the distinct critical gaps used in our macroscopic formulation ($g_{crit}^{CAV} \approx 10$ m vs. $g_{crit}^{reg} \approx 30$ m) emerge endogenously from rational utility maximization, rather than being arbitrary inputs.

We simulated a rational Ego agent attempting to merge into a target lane occupied by either a cooperative CAV or a non-cooperative HDV. The agent's decision logic was governed by the Cross-Entropy Method (CEM) optimizer with a utility function weighing speed ($w_v=4.0$) against safety ($w_s=2.5$) and effort ($w_a=0.1$).

Figure \ref{fig:micro_validation} presents the resulting decision boundaries. The simulation reveals a behavioral asymmetry driven by the opponent's type:

\begin{itemize}
    \item Exploitation of Cooperation (green circle marks): Against a CAV, the rational agent accepts gaps as small as 12 m. The solver correctly identifies that the cooperative opponent will brake to resolve the conflict, rendering the proximity penalty transient (lasting only 1--2 seconds). The agent effectively forced its way for the merge, taking advantage of the CAV's collision avoidance logic.
    
    \item Rejection of Conflict (red cross marks): Against an HDV, the same agent rejects these small gaps and requires a significantly larger opening ($\approx$ 23 m) to merge. Because the HDV maintains speed, any cut-in at a smaller gap results in a persistent proximity penalty that accumulates over the entire prediction horizon (10 s), outweighing the utility of the speed gain.
\end{itemize}

Although the critical gap for HDVs derived from the game-theoretic model ($\approx$ 23 m) is slightly more aggressive than the standard safety gap utilized in our macroscopic configuration (30 m), the structural finding remains robust: the acceptable gap against a CAV is roughly half that required against an HDV. This fundamental disparity confirms that mixed traffic cannot be accurately captured by a single interaction parameter. The identification of this ``exploitation window'', specifically the range $[12\text{ m}, 23\text{ m}]$ where a strategic agent will cut off a CAV but yield to an HDV. It provides the physical justification for the dual-threshold mechanism integrated into the macroscopic simulation presented in the following section.

\begin{figure}[th!]
    \centering
    \includegraphics[width=\linewidth]{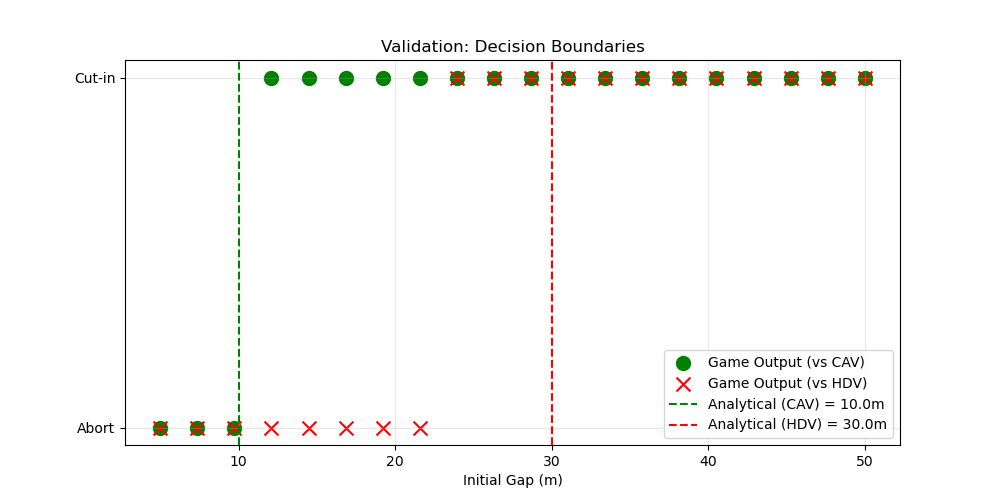} 
    \caption{Validation of HDV's Decision Boundaries in the Cut-in Game: Decision boundaries derived from the CEM game solver. The agent exploits the cooperative nature of CAVs (accepting 12 m gaps) while respecting the stubbornness of HDVs (requiring 23 m gaps).}
    \label{fig:micro_validation}
\end{figure}

\subsubsection{Friction-Induced Compression}
\label{sec:density_nucleation}

Having established the microscopic interaction parameters, we next isolate the macroscopic consequences of these interactions. Specifically, we seek to verify the existence of the friction-induced compression, which theoretically predicts that the game-theoretic source term triggers spontaneous density nucleation even in the absence of geometric bottlenecks.

To test this, we simulated a scenario with constant free-flow inflow ($\rho_{in} = 0.015$ veh/m) and no physical bottlenecks (i.e., constant capacity). In a standard LWR model, the density profile would remain flat and invariant over time. However, we activated the game-theoretic friction by having 50\% of CAV ($m_C=0.5$) within the weaving zone ($x \in [500, 1000]$ m).

Figure \ref{fig:friction_pump} illustrates the temporal evolution of the average density within the weaving segment. Despite the constant inflow (black dashed line), the internal density (red solid line) spontaneously rises and stabilizes at a higher equilibrium value ($\rho \approx 0.0177$ veh/m). This 18\% increase in density represents pure behavioral compression. Physically, the strategic cut-ins act as a distributed resistance field, forcing the traffic stream to compress to maintain mass conservation. This result empirically confirms that the game-theoretic interaction term functions as a ``behavioral bottleneck,'' inducing congestion solely through the heterogeneity of driver strategies.

\begin{figure}[th!]
    \centering
    \includegraphics[width=\linewidth]{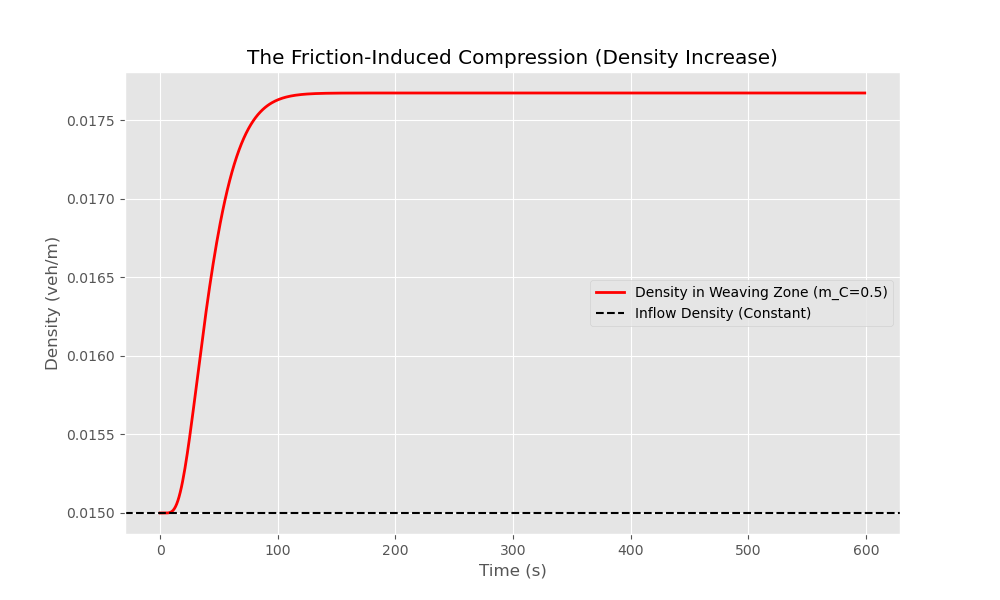}
    \caption{Friction-Induced Compression: Temporal evolution of traffic density within the weaving zone. Despite constant inflow conditions (black dashed line), the activation of game-theoretic friction ($m_C=0.5$) causes the local density to spontaneously rise and stabilize at a higher equilibrium (red solid line), confirming the existence of the behavioral bottleneck effect.}
    \label{fig:friction_pump}
\end{figure}

\subsubsection{Capacity Decrease}
\label{sec:capacity_drop}

While the previous section demonstrated that behavioral friction increases density, a more critical question for traffic operations is whether this friction reduces the maximum throughput (capacity) of the facility. To quantify this, we first reconstructed the equilibrium FD for the worst-case interaction scenario ($m_C=50\%$). Consistent with the formulation in \cite{jin2010kinematic}, we utilized the Del Castillo and Benitez \cite{del1995functional} functional form for the equilibrium speed-density relationship.

Figure \ref{fig:fd_shift} compares the resulting flow-density curves for the Baseline ($m_C=0\%$, blue) and the Strategic Cut-In scenario ($m_C=50\%$, red). The results reveal a significant macroscopic capacity degradation driven by microscopic exploitation. The active exploitation of CAVs amplifies the effective density, causing a downward suppression of the FD. The capacity drops from approximately 1400 veh/h to 1300 veh/h ($\approx 7\%$ reduction).

\begin{figure}[th!]
    \centering
    \includegraphics[width=\linewidth]{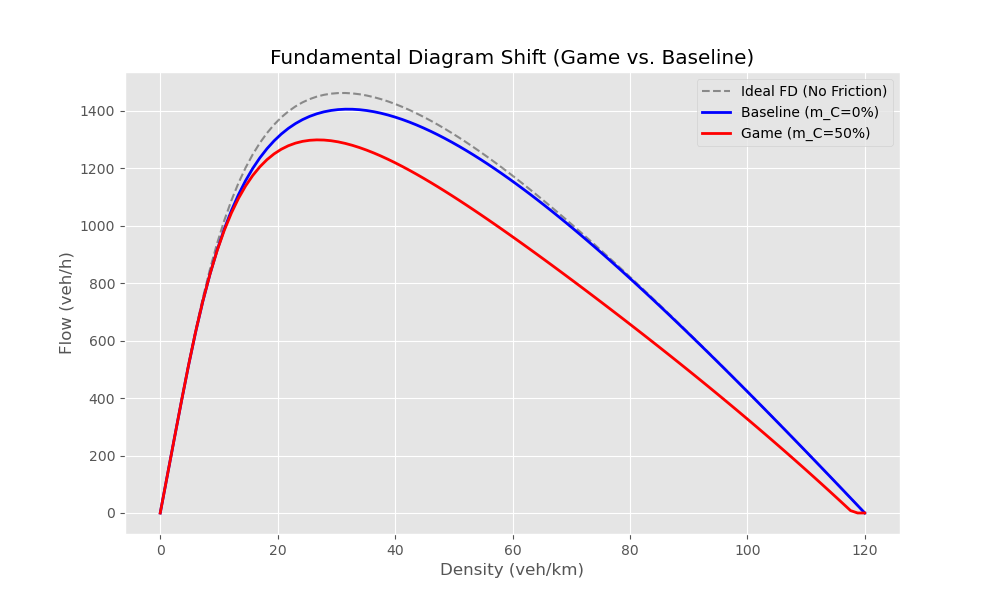}
    \caption{FD Shift: Comparison of equilibrium flow-density relationships. The active exploitation of CAVs ($m_C=50\%$, red line) suppresses the flow curve compared to the Baseline ($m_C=0\%$, blue line), resulting in capacity degradation.}
    \label{fig:fd_shift}
\end{figure}

To investigate how this degradation varies with technology adoption, we performed a sensitivity analysis by sweeping the CAV penetration rate $m_C$ from 0\% to 100\%. Figure \ref{fig:capacity_sweep} presents the resulting capacity curve, revealing a strictly convex relationship between penetration rate and system capacity. The degradation follows an inverted parabolic trend, closely mirroring the interaction frequency term $m_C(1-m_C)$ in our friction model. In the initial phase ($0\% < m_C < 45\%$), the introduction of CAVs creates a growing number of exploitable targets. Human drivers encounter frequent opportunities to cut in, thereby increasing the aggregate friction until the system reaches its minimum efficiency at approximately $m_C \approx 45\%$, where the frequency of HDV-CAV interactions is maximized. Beyond this inflection point ($m_C > 45\%$), a recovery phase emerges. As CAVs become dominant, the population of aggressive HDVs ($m_H$) dwindles; although targets remain plentiful, there are fewer aggressors to exploit them. Consequently, the capacity gradually recovers, ultimately exceeding the baseline only when $m_C > 85\%$.

\begin{figure}[th!]
    \centering
    \includegraphics[width=\linewidth]{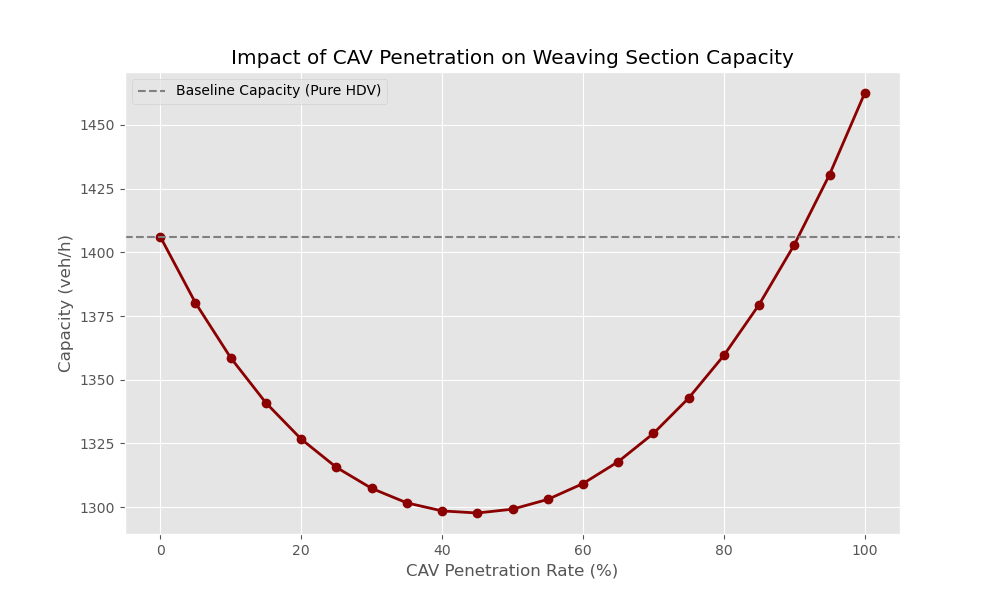}
    \caption{Convex Capacity Curve: Sensitivity of macroscopic capacity to CAV penetration rate. The system exhibits a non-linear, convex response, where capacity initially degrades due to the high interaction frequency between aggressive HDVs and passive CAVs. The minimum capacity (``valley of efficiency'') occurs near 45\% penetration, and performance only exceeds the pure-HDV baseline at very high adoption rates ($>85\%$).}
    \label{fig:capacity_sweep}
\end{figure}

\subsubsection{Spatiotemporal Instability}
\label{sec:shockwaves}

Finally, we examine the dynamic impact of game-theoretic friction on the propagation of congestion. Static capacity drops are dangerous not only because they limit throughput, but because they alter the wave speeds governing traffic breakdown. To visualize this, we simulated a high-demand scenario ($\rho_{in} > \rho_{crit}$) where the weaving section acts as an active bottleneck, triggering a backward-moving queue.

Figure \ref{fig:shockwaves} presents the spatiotemporal density evolution for the Baseline (top) and Strategic Cut-In (bottom) scenarios. Both cases exhibit the formation of a shockwave upstream of the weaving zone ($x < 500$ m), separating the free-flow state (dark blue) from the congested state (cyan). However, the trajectories of these shockwaves differ fundamentally.

In the Game-Theoretic scenario, the reduced discharge capacity leads to a larger flux deficit ($Q_{in} - Q_{cap}$). According to the Rankine-Hugoniot condition ($w = \Delta Q / \Delta \rho$), this deficit dictates the speed of the backward shockwave. As seen in the figure, the shockwave front in the Strategic scenario is significantly flatter, indicating a higher upstream propagation speed. The queue reaches the upstream boundary ($x=0$ m) at approximately $t=400$ s, whereas in the Baseline scenario, the queue propagates more slowly, reaching the boundary nearly 30\% later ($t \approx 600$ s).

This validates the theoretical prediction of instability amplification: the game-theoretic friction does not only constrain the capacity, it also induces an shockwave that propagates faster and affects a larger portion of the upstream network. This rapid back-propagation implies that mixed-autonomy congestion will be inherently more difficult to dissipate than traditional recurring congestion.

\begin{figure}[th!]
    \centering
    \includegraphics[width=\linewidth]{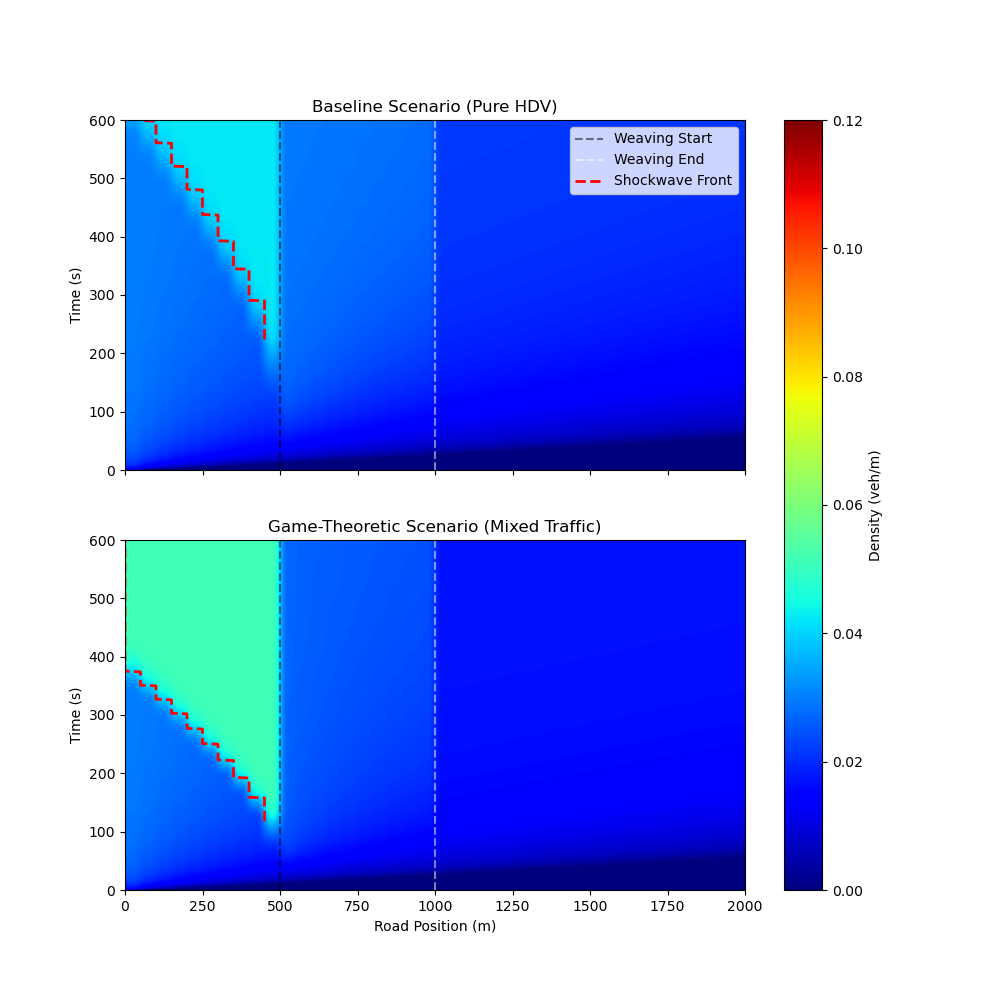}
    \caption{Spatiotemporal Instability Under High Demand: Evolution of traffic density in the $x$-$t$ plane with a constant inflow $\rho_{in} = 0.03$ veh/m. While the Baseline (top) exhibits a transition to a boundedly stable congested state, the Game-Theoretic scenario (bottom) triggers a severe capacity drop, resulting in a high-density jam (indicated by the green region, $\rho \approx 0.05$ veh/m). Consequently, the backward recovery shockwave in the Game-Theoretic model propagates significantly faster, reaching the upstream boundary ($x=0$) approximately 200 seconds earlier than the Baseline, demonstrating the impact of negotiation friction on bottleneck capacity.}
    \label{fig:shockwaves}
\end{figure}

\subsubsection{Parameter Sensitivity and Robustness}
\label{sec:sensitivity}

A key parameter in our macroscopic formulation is the lane change impact duration ($t_{LC}^{\text{cut}}$). In Jin's original kinematic wave formulation \cite{jin2010kinematic}, the corresponding parameter $\tau$ is defined strictly as the physical maneuver duration: the time a vehicle occupies space in both lanes (typically 2-4 s). However, in the context of mixed-autonomy conflicts, aggressive cut-ins induce strong braking reactions that persist long after the physical maneuver is complete. Consequently, we extend the definition of $t_{LC}^{\text{cut}}$ to represent the effective impact duration, capturing both the maneuver itself and the subsequent relaxation time required for the follower to recover their desired headway.

In the results above, we assumed a representative value of $t_{LC}^{\text{cut}} = 20$ s, corresponding to the time required for such a disturbance to dissipate in heavy traffic. To ensure our findings are not artifacts of this specific parameter choice, we performed a robustness check by varying $t_{LC}^{\text{cut}}$ from 5 s to 25 s. Figure \ref{fig:robustness} illustrates the sensitivity of the capacity degradation to this parameter. The relationship is approximately linear: while the magnitude of the capacity drop varies (ranging from $\approx 1\%$ at 5 s to $\approx 9\%$ at 25 s), the phenomenon itself is robust. Even with a conservative estimate of 10 s, a statistically significant capacity reduction persists. This confirms that the behavioral bottleneck is a structural consequence of the game-theoretic interaction, not merely a result of parameter tuning.

\begin{figure}[h!]
    \centering
    \includegraphics[width=\linewidth]{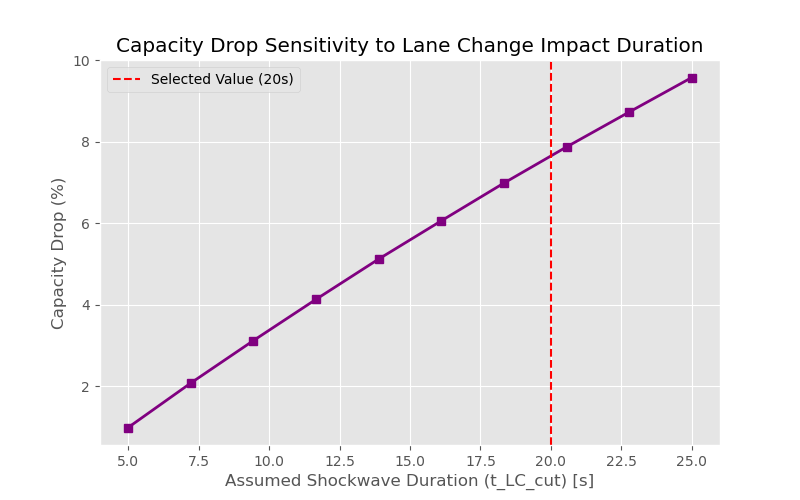}
    \caption{Robustness Check: Sensitivity of capacity degradation to the assumed lane-change impact duration ($t_{LC}^{\text{cut}}$). While the magnitude of the drop scales linearly with the persistence of the friction, the degradation remains positive and significant across the entire physically plausible range.}
    \label{fig:robustness}
\end{figure}

\section{Discussion}
\label{sec:discussion}

The simulation results presented in Section \ref{sec:simulation} confirm that mixed-autonomy traffic is susceptible to a unique form of instability driven not by vehicle dynamics, but by game-theoretic exploitation. The core contribution of this work is the identification of friction-induced compression, a phenomenon where congestion arises purely from behavioral friction rather than geometric bottlenecks. As verified in our microscopic validation (Section \ref{sec:micro_validation}), rational agents identify an ``exploitation window'' (gaps between 12 m and 23 m) where they can force a merge against a risk-averse CAV but would be rejected by a stubborn HDV. Macroscopically, this manifests as a behavioral bottleneck, where the game-theoretic interaction term acts as a distributed resistance field. Physically, every successful exploitative cut-in forces the following CAV to brake, creating a localized void that is immediately filled by upstream pressure, resulting in a permanent increase in local density. This validates Theorem \ref{thm:instability} and explains why adding ``efficient'' but passive CAVs can paradoxically nucleate congestion in weaving zones.

This behavioral friction does not affect the system linearly; rather, our sensitivity analysis reveals a convex relationship where system performance is minimized at intermediate penetration rates ($m_C \approx 45\%$). This aligns with the theoretical interaction term $m_C(1-m_C)$, confirming that instability is maximized when the frequency of heterogeneous interactions peaks. This finding challenges the prevailing optimism that low-level CAV deployment will yield immediate flow benefits. Instead, it suggests that the transition period will be characterized by a regime of instability. As long as a significant population of aggressive human drivers remains, they will treat the growing population of CAVs as exploitable infrastructure. System performance only recovers when CAVs achieve saturation ($>85\%$), sufficiently diluting the aggressive population to suppress the friction term.

Furthermore, the impact of this friction extends beyond static capacity degradation to dynamic spatiotemporal instability. The formation of faster-propagating backward shockwaves, as observed in Section \ref{sec:shockwaves}, implies that game-theoretic congestion is kinematically more robust than recurrent congestion. Because the jam propagates more rapidly upstream (due to the larger flux deficit defined by the Rankine-Hugoniot condition), it affects a larger portion of the network in less time. This rapid back-propagation necessitates a fundamental shift in CAV control logic from purely ``reactive safety'' to ``strategic defense''. Current algorithms designed for string stability are, in a game-theoretic sense, dominated strategies that encourage exploitation.

To mitigate these effects, future CAV control strategies must move beyond unconditional cooperation. Effective countermeasures may require the adoption of defensive game theory, where CAVs dynamically tighten gaps in high-friction zones to physically close the exploitation window ($< 12$ m). While our current model assumes a homogeneous population of rational HDVs, the robustness of these findings across varying impact durations (Section \ref{sec:sensitivity}) suggests that the behavioral bottleneck is a structural reality of mixed traffic. Ultimately, maximizing the social utility of autonomous transportation may require algorithms that occasionally reject individual cut-ins to preserve the aggregate flow, prioritizing system stability over local courtesy.

\section{Conclusion}
\label{sec:conclusion}

This study has established a rigorous analytical link between microscopic game-theoretic interactions and macroscopic traffic flow stability in mixed-autonomy environments. By modeling the HDV cut-in maneuver as a Stackelberg game against risk-averse CAVs, we identified a distinct ``exploitation window'' where human drivers leverage the defensive programming of automated systems to force merges at substandard gaps. We successfully bridged this behavioral asymmetry to the continuum scale by deriving a probabilistic friction term ($\epsilon$) that endogenously modifies the kinematic wave equation. This novel micro-macro framework moves beyond static heterogeneity parameters, allowing the macroscopic model to dynamically capture the capacity degradation driven specifically by the frequency and aggressiveness of strategic exploitations.

Our theoretical and numerical analyses reveal that this behavioral friction acts as a deterministic destabilizer of the traffic stream. We proved that the gradients of game-induced friction generate positive perturbation source terms, triggering phantom jams and premature congestion even in the absence of geometric bottlenecks. Most critically, we demonstrated a convex relationship between CAV market penetration and road capacity, identifying a ``valley of efficiency'' at intermediate penetration rates (approximately 45\%). In this regime, the system suffers maximum instability because the frequency of exploitable HDV-CAV pairings is maximized, confirming that the defensive nature of early-deployment CAVs may paradoxically degrade network throughput before sufficient saturation is reached to realize cooperative benefits.

Future research will extend this framework to closed-loop control strategies, exploring how CAVs can adapt their game-theoretic parameters to mitigate exploitation. A promising direction is the development of more ``socially-aware'' CAV policies that balance defensive safety with strategic assertiveness to close the exploitation window. Additionally, we aim to validate these findings using high-fidelity trajectory data from real-world mixed-traffic pilots.  By calibrating the game-theoretic utility functions against empirical cut-in distributions, we can further refine the friction model to support the design of traffic management strategies that specifically target behavioral bottlenecks in the transition era of automated driving.

\bibliography{refs}

\appendix
\section*{Appendix}
\subsection*{Solving the HDV Cut-in Control Problem}

The core idea is to transform the continuous optimization problem into a discrete one that can be solved with a numerical algorithm. This is done by discretizing the planning horizon and searching for the best sequence of control inputs.

The total planning horizon is denoted as \(T\), and we can discretize it by dividing it into \(N\) discrete steps, with \(\Delta t = T/N\). The HDV will plan a sequence of \(N\) actions, one for each step.

The decision variable for the optimization is a sequence of control inputs for the HDV over the planning horizon:
\begin{equation}
\mathbf{U}_\text{HDV} = [\mathbf{u}_\text{HDV}(t_0), \mathbf{u}_\text{HDV}(t_1), ..., \mathbf{u}_\text{HDV}(t_{N-1})]
\end{equation}
where \(\mathbf{u}_\text{HDV}(t_k) = [u_{x,k}, u_{y,k}]^T\) is a control input from the action space \(\mathcal{A}_\text{HDV}\).

For any given action sequence, \(\mathbf{U}_\text{HDV}\), the HDV predicts the resulting future states of the system by simulating the dynamics forward in time, step by step, for \(k = 0, ..., N-1\). The Ego HDV state is updated by:
\begin{equation}
\mathbf{x}_\text{HDV} (t_{k+1}) = F(\mathbf{x}_\text{HDV}(t_k), \mathbf{u}_\text{HDV}(t_k))
\end{equation}
and the Target Vehicle (TV) state is updated using the type-dependent policy \(\pi_\text{TV}\):
\begin{equation}
\mathbf{u}_\text{TV} (t_{k}) = \pi_\text{TV}(\mathbf{s}(t_k))
\end{equation}
\begin{equation}
\mathbf{x}_\text{TV} (t_{k+1}) = F(\mathbf{x}_\text{TV}(t_k), \mathbf{u}_\text{TV}(t_k))
\end{equation}
where \(F(\cdot)\) is the state transition function based on vehicle dynamics. This process generates a sequence of system states: \([\mathbf{s}(t_0), \mathbf{s}(t_1), ..., \mathbf{s}(t_N)]\).

For the predicted trajectory, we can calculate the total discounted utility for the HDV:
\begin{equation}
U_\text{HDV}(\mathbf{U}_\text{HDV}) = \sum_{k=0}^{N-1}{ \gamma^k \cdot r_\text{HDV} (\mathbf{s}(t_k), \mathbf{u}_\text{HDV}(t_k))}
\end{equation}
The optimization problem is now to find the action sequence, \(\mathbf{U}_\text{HDV}^*\), that maximizes the total utility function:
\begin{equation}
\mathbf{U}_\text{HDV}^* = \arg\max_{\mathbf{U}_\text{HDV}} U_\text{HDV}(\mathbf{U}_\text{HDV})
\end{equation}
Since the action space is continuous, this search requires a numerical optimization algorithm, and we choose the Cross-Entropy Method (CEM), a derivative-free technique. The application of CEM on the HDV trajectory optimization is shown in Algorithm \ref{alg:cem}. The result of the optimization is the entire optimal sequence of actions, \(\mathbf{U}_\text{HDV}^* = [\mathbf{u}_\text{HDV}^*(t_0), \mathbf{u}_\text{HDV}^*(t_1), ..., \mathbf{u}_\text{HDV}^*(t_N-1)]\). The HDV only executes the first action in this sequence, \(\mathbf{u}_\text{HDV}^*(t_0)\), and at the next time step, the entire process is repeated with new perceived information. The final optimal control input is:
\begin{equation}
( \mathbf{u}_x^*(t), \mathbf{u}_y^*(t) ) = \mathbf{u}_\text{HDV}^*(t_0)
\end{equation}

\begin{algorithm}
\caption{CEM for HDV Trajectory Optimization}
\label{alg:cem}
\begin{algorithmic}[1]
\State \textbf{Input:} Initial state $\mathbf{s}(t_0)$, planning horizon $N$, samples $K$, elite samples $M$, iterations $I$.
\State \textbf{Initialize:} Mean action sequence $\boldsymbol{\mu} \leftarrow \mathbf{0}$; Standard deviation sequence $\boldsymbol{\sigma} \leftarrow \mathbf{1}$.
\For{$i = 1$ to $I$}
    \State \textit{// Sample K action sequences}
    \State Sample $\mathbf{U}_j \sim \mathcal{N}(\boldsymbol{\mu}, \text{diag}(\boldsymbol{\sigma}^2))$ for $j = 1, ..., K$.
    \State
    \State \textit{// Evaluate each sample's utility via simulation rollout}
    \For{$j = 1$ to $K$}
        \State Calculate total utility $U_j \leftarrow U_{\text{HDV}}(\mathbf{s}(t_0), \mathbf{U}_j)$.
    \EndFor
    \State
    \State \textit{// Select the M elite samples}
    \State Let $\mathcal{E}$ be the set of the top $M$ action sequences $\{\mathbf{U}_j\}$ with the highest utilities.
    \State
    \State \textit{// Update the distribution to fit the elite set}
    \State $\boldsymbol{\mu} \leftarrow \frac{1}{M} \sum_{\mathbf{U} \in \mathcal{E}} \mathbf{U}$
    \State $\boldsymbol{\sigma} \leftarrow \sqrt{\frac{1}{M} \sum_{\mathbf{U} \in \mathcal{E}} (\mathbf{U} - \boldsymbol{\mu})^2}$
\EndFor
\State \textbf{Return:} The first action from the final mean sequence, $\mathbf{u}^*(t_0) \leftarrow \boldsymbol{\mu}[0]$.
\end{algorithmic}
\end{algorithm}

\end{document}